%% file: main.tex
\documentclass{article}
\PassOptionsToPackage{numbers, sort, compress}{natbib}
\usepackage{fullpage}
\usepackage{tikz}
\usepackage[utf8]{inputenc} % allow utf-8 input
\usepackage[T1]{fontenc}    % use 8-bit T1 fonts
\usepackage{hyperref}       % hyperlinks
\usepackage{url}            % simple URL typesetting
\usepackage{booktabs}       % professional-quality tables
\usepackage{amsfonts}       % blackboard math symbols
\usepackage{nicefrac}       % compact symbols for 1/2, etc.
\usepackage{microtype}      % microtypography
\usepackage{xcolor}         % colors

\usepackage{microtype}
\usepackage{graphicx}
\usepackage{lipsum}
\usepackage{float}
\usepackage{array}
\usepackage{xspace}
\usepackage{multirow}
\usepackage{algorithm, algpseudocode}
\usepackage{stfloats}
\usepackage{natbib}
\usepackage{wrapfig}
\usepackage{subfig}
\usepackage{bbm}
\usepackage{caption}
\usepackage{bm}
\usepackage{verbatim}
\usepackage{enumitem}
\usepackage{amsfonts, amssymb, amsthm, amsmath, thm-restate, thmtools} %%

\input{macro}

\title{Incentivizing Truthful Submissions in a Data Marketplace for Mean Estimation}

\author{%
  Keran Chen \\
  UW-Madison\\
  \texttt{kchen429@wisc.edu} 
   \and
   Alex Clinton \\
   UW-Madison \\
   \texttt{aclinton@wisc.edu} 
   \and
   Kirthevasan Kandasamy \\
   UW-Madison \\ \texttt{kandasamy@cs.wisc.edu} \\  
}
\date{}

\begin{document}

\maketitle

\input{section/table}
\input{section/1abstract}
\input{section/2introduction}

\input{section/3setup}

\input{section/hardness}
\input{section/4cost_known}

\input{section/5profit_maximization}

\input{section/6conclusion}

\bibliographystyle{apalike}
\bibliography{ref}

\clearpage

\clearpage
\appendix
\thispagestyle{empty}

\onecolumn
\appendix

\input{Appendices/mean_estimation}

\input{Appendices/app_cost_known}

\input{Appendices/app_technical}

\input{Appendices/app_hardness_result}
\vspace{-10pt}

\input{Appendices/app_profit_max}

\newpage
\input{Appendices/notation_table}

\end{document}

%% file: macro.tex
\allowdisplaybreaks

\theoremstyle{remark}

\theoremstyle{definition}

\theoremstyle{plain}

\newcommand{\welf}{W}

\newcommand{\mech}{M}
\newcommand{\mechopt}{\mech}

\newcommand{\initdata}{Z}
\newcommand{\initdataii}[1]{\initdata_{#1}}
\newcommand{\initdatai}{\initdataii{i}}
\newcommand{\initdataj}{\initdataii{j}}

\newcommand{\receiveddata}{Y'}

\newcommand{\subdata}{X'}
\newcommand{\subdataii}[1]{\subdata_{#1}}
\newcommand{\subdatai}{\subdataii{i}}
\newcommand{\subdatami}{\subdataii{-i}}

\newcommand{\val}{v}
\newcommand{\valii}[1]{\val_{#1}}
\newcommand{\vali}{\valii{i}}
\newcommand{\valj}{\valii{j}}

\newcommand{\valitem}{v}
\newcommand{\valitemii}[1]{\valitem_{#1}}
\newcommand{\valitemi}{\valitemii{i}}
\newcommand{\valitemj}{\valitemii{j}}

\newcommand{\valfuncnob}{v^{\rm e}}
\newcommand{\valfunc}{v^{\rm e}}
\newcommand{\valfuncii}[1]{\valfunc_{#1}}

\newcommand{\valfuncj}{\valfuncii{j}}

\newcommand{\valmech}{v^{\rm B}}
\newcommand{\valmechii}[1]{\valmech_{#1}}

\newcommand{\valmechj}{\valmechii{j}}

\newcommand{\valdata}{v^{\rm iid}}
\newcommand{\valdataii}[1]{\valdata_{#1}}
\newcommand{\valdatai}{\valdataii{i}}
\newcommand{\valdataj}{\valdataii{j}}

\newcommand{\itemprice}{q}

\newcommand{\mjp}{\overline{m}_{j}}

\newcommand{\optitemprice}{q^{\rm OPT}}

\newcommand{\efscheme}{M}

\newcommand{\efprofit}{{\rm profit}}
\newcommand{\optnone}{\reqamnt_1^{\rm OPT}}

\newcommand{\OPTREVmech}{\mathrm{rev}^{\rm OPT}}
\newcommand{\ExpectREV}{\overline{\mathrm{rev}}}
\newcommand{\OPTREV}{\overline{\mathrm{rev}}^{\rm OPT}}

\newcommand{\utilb}{u^{\rm B}}
\newcommand{\utilbii}[1]{\utilb_{#1}}

\newcommand{\utilbj}{\utilbii{j}}

\newcommand{\utilc}{u^{\rm C}}
\newcommand{\utilcii}[1]{\utilc_{#1}}
\newcommand{\utilci}{\utilcii{i}}

\newcommand{\identity}{\II}

\newcommand{\subfunc}{X}
\newcommand{\subfunci}{\subfunc_i}
\newcommand{\subfuncj}{\subfunc_j}

\newcommand{\nni}{n_i}

\newcommand{\colamnt}{n}
\newcommand{\colamnti}{\colamnt_i}

\newcommand{\colamntr}{\colamnt'}
\newcommand{\colamntri}{\colamntr_i}

\newcommand{\reqamnt}{N}

\newcommand{\reqamnti}{\reqamnt_i}

\newcommand{\reqamntr}{\reqamnt'}
\newcommand{\reqamntri}{\reqamntr_i}
\newcommand{\reqamntrmi}{\reqamntr_{-i}}
\newcommand{\reqamntrj}{\reqamntr_j}

\newcommand{\numdata}{N^{\rm tot}}

\newcommand{\parahead}[1]{\vspace{0.025in}\noindent\textbf{#1. }}
\newcommand{\subparahead}[1]{\vspace{0.01in}\noindent\emph{#1. }}
\newcommand{\paraheadwnogap}[1]{\vspace{0.0in}\noindent\textbf{#1. }}
\newcommand{\subparaheadwnogap}[1]{\vspace{0.0in}\noindent\emph{#1. }}

\RequirePackage{amsmath}
\RequirePackage{amssymb}
\RequirePackage{mathtools}
\ifx\proof\undefined 
\RequirePackage{amsthm}
\fi
\RequirePackage{bm} 
\RequirePackage{url}

\newcommand{\eps}{\epsilon}

\newcommand{\Bcal}{\mathcal{B}}
\newcommand{\Ccal}{\mathcal{C}}

\newcommand{\Ncal}{\mathcal{N}}

\newcommand{\Xcal}{\mathcal{X}}

\newcommand{\EE}{\mathbb{E}} % Expectation
\newcommand{\II}{\mathbb{I}} % Indicator
\newcommand{\NN}{\mathbb{N}} % Natural numbers
\newcommand{\RR}{\mathbb{R}} % Real numbers
\newcommand*{\argmax}{\mathop{\mathrm{argmax}}}

\newcommand{\bigO}{\mathcal{O}}

\newcommand{\bigOtilde}{\widetilde{\mathcal{O}}}
\newcommand{\defeq}{:=}

\newcommand{\rbr}[1]{\left(#1\right)}
\newcommand{\sbr}[1]{\left[#1\right]}
\newcommand{\cbr}[1]{\left\{#1\right\}}

\newcommand{\abr}[1]{\left|#1\right|}

\ifx\BlackBox\undefined
\newcommand{\BlackBox}{\rule{1.5ex}{1.5ex}}  % end of proof
\fi

\ifx\QED\undefined
\def\QED{~\rule[-1pt]{5pt}{5pt}\par\medskip}
\fi

\ifx\proof\undefined
\newenvironment{proof}{\par\noindent{\bf Proof\ }}{\hfill\BlackBox\\[2mm]}
\fi
\ifx\proofsketch\undefined
\newenvironment{proofsketch}{\par\noindent{\bf Proof Sketch\ }}{\hfill\BlackBox\\[2mm]}
\fi

\ifx\theorem\undefined
\newtheorem{theorem}{Theorem}
\fi
\ifx\example\undefined
\newtheorem{example}{Example}
\fi
\ifx\property\undefined
\newtheorem{property}{Property}
\fi
\ifx\definition\undefined
\newtheorem{definition}{Definition}
\fi
\ifx\remark\undefined
\newtheorem{remark}{Remark}
\fi
\ifx\assumption\undefined
\newtheorem{assumption}{Assumption}
\fi
\ifx\lemma\undefined
\newtheorem{lemma}[theorem]{Lemma}
\fi
\ifx\proposition\undefined
\newtheorem{proposition}[theorem]{Proposition}
\fi
\ifx\corollary\undefined
\newtheorem{corollary}[theorem]{Corollary}
\fi
\ifx\conjecture\undefined
\newtheorem{conjecture}[theorem]{Conjecture}
\fi
\ifx\fact\undefined
\newtheorem{fact}[theorem]{Fact}
\fi
\ifx\claim\undefined
\newtheorem{claim}[theorem]{Claim}
\fi
\ifx\assum\undefined

\fi
\ifx\condition\undefined
\newtheorem{condition}[theorem]{Condition}
\fi

\definecolor{darkblue}{rgb}{0.0,0.0,0.7} % for hyper-links
\hypersetup{colorlinks,breaklinks, linkcolor=darkblue, urlcolor=darkblue,
            anchorcolor=darkblue, citecolor=darkblue}
\newcommand\numberthis{\addtocounter{equation}{1}\tag{\theequation}}

\makeatletter
\DeclareRobustCommand\widecheck[1]{{\mathpalette\@widecheck{#1}}}
\def\@widecheck#1#2{%
    \setbox\z@\hbox{\m@th$#1#2$}%
    \setbox\tw@\hbox{\m@th$#1%
       \widehat{%
          \vrule\@width\z@\@height\ht\z@
          \vrule\@height\z@\@width\wd\z@}$}%
    \dp\tw@-\ht\z@
    \@tempdima\ht\z@ \advance\@tempdima2\ht\tw@ \divide\@tempdima\thr@@
    \setbox\tw@\hbox{%
       \raise\@tempdima\hbox{\scalebox{1}[-1]{\lower\@tempdima\box
\tw@}}}%
    {\ooalign{\box\tw@ \cr \box\z@}}}
\makeatother

\newcommand{\contributors}{\Ccal}
\newcommand{\buyers}{\Bcal}

\newcommand{\cost}{c}
\newcommand{\costi}{\cost_i}
\newcommand{\costj}{\cost_j}
\newcommand{\costone}{\cost_1}

\newcommand{\muhat}{\widehat{\mu}}
\newcommand{\muhatsm}{\widehat{\mu}}

\newcommand{\mm}{m}
\newcommand{\mmj}{\mm_j}
\newcommand{\mmr}{\mm'}
\newcommand{\mmrj}{\mm'_j}

\newcommand{\mmoptj}{\mm_j}
\newcommand{\mmoptk}{\mm_k}

\newcommand{\dataspace}{\Xcal}

\newcommand{\ie}{\textit{i}.\textit{e}., }
\newcommand{\eg}{\textit{e}.\textit{g}., }

\newcommand{\price}{p}

\newcommand{\pricej}{\price_j}
\newcommand{\priceoptj}{\price_j}
\newcommand{\priceoptk}{\price_k}

\newcommand{\pricer}{p'}
\newcommand{\priceri}{\pricer_i}
\newcommand{\pricerj}{\pricer_j}

\newcommand{\pricecurv}{q} 
\newcommand{\pay}{\pi}

\newcommand{\payi}{\pay_i}
\newcommand{\payopti}{\pay_i}

\newcommand{\payr}{\pay'}
\newcommand{\payri}{\pay'_i}

\newcommand{\mmk}{m_k}

\newcommand{\mmNi}{m_{N,i}}
\newcommand{\mmNj}{m_{N,j}}

\newcommand{\optprofitA}{{\rm OPT}^+}

\newcommand{\numdataA}{N^{+}}

\newcommand{\pricejA}{\itemprice_j^{+}}
\newcommand{\pricekA}{\itemprice_k^{+}}

\newcommand{\mmjA}{m_{j}^{+}}
\newcommand{\mmkA}{m_{k}^{+}}
\newcommand{\strat}{s}
\newcommand{\strati}{\strat_i}
\newcommand{\stratmi}{\strat_{-i}}

\newcommand{\stratopt}{s^{\identity}}
\newcommand{\stratopti}{\stratopt_i}
\newcommand{\stratoptmi}{\stratopt_{-i}}

\newcommand{\nnopt}{\reqamnt}
\newcommand{\nnopti}{\nnopt_i}

\newcommand{\profit}{{\rm profit}}

\newcommand{\blprofit}{{\rm OPT}}
\newcommand{\OPT}{{\rm OPT}}
\newcommand{\inputOPT}{\widetilde{\rm OPT}}
\newcommand{\inputreqamnt}{N^{\widetilde{\rm OPT}}}
\newcommand{\EEV}[1]{\mathop{\EE}_{#1}}

\newcommand{\efclass}{{\rm\bf EF}}

\newcommand{\mutilde}{\widetilde{\mu}}

\newcommand{\reqamntoneopt}{N^{\OPT}}

\newcommand{\buyersell}{\widetilde{m}}
\newcommand{\buyersellj}{\buyersell_j}

\newcommand{\buyerexpprice}{\widetilde{p}}
\newcommand{\buyerexppricej}{\buyerexpprice_j}

%% file: section/table.tex
\newcommand{\insertTableApproxResults}{%
\begin{table*}[t]
\caption{Algorithms from Prior Work}
    \centering
    \renewcommand{\arraystretch}{1.6}
   \begin{tabular}{lll}
    \hline
    Algorithm & Assumptions on buyer valuations & Time complexity  \\ \hline
    \citet{hartline2005near} & --  & $\bigOtilde(2^N \eps^{-N})$    \\ \hline
    \citet{chawla2022pricing} & Monotonicity & $N^{\bigO\rbr{\eps^{-2}\log\eps^{-1}}}$  \\ \hline
    \multirow{3}{*}{\citet{chen2024learning}} 
    & Monotonicity, $K$ types & $\bigOtilde(N^K \eps^{-K}) $ \\ \cline{2-3}
    & Monotonicity, $K$ types, smoothness & $\bigOtilde\left( \eps^{-2K}  \right)$ \\ \cline{2-3}
    & Monotonicity, $K$ types, diminishing return & $\bigOtilde( \eps^{-3K}\log^K(N) )$ 
    \\ \hline
\end{tabular}
    \vspace{0.05in}
    \caption*{\footnotesize Table 1: Algorithms from prior work applicable for pricing ordered items, along with their assumptions on buyer valuations and time complexities to obtain an $|\buyers|\bigO(\eps)$-optimal solution when there are $N$ total points.
    \citet{chawla2022pricing} were the first to study this problem, assuming monotonic buyer valuations.
    \citet{chen2024learning} studied this problem in the context of data pricing, introducing additional assumptions on buyer valuations, such as finite buyer types, smoothness, and
    diminishing returns. 
    }\label{tb:approxresults}
\end{table*}
}

\newcommand{\insertOIPAlgo}{%

\begin{algorithm}[t]
    \caption{\label{alg:oipalgo}}
    \begin{algorithmic}[1]
    \State {\bfseries Input:} Buyer valuations $\{\valj\}_{j\in\buyers}$, contributor costs
    $\{\costi\}_{i\in\contributors}$, 
    an algorithm for ordered item pricing $A$, 
    approximation parameter $\eps$.
    \State \textbf{for} $N\in \{1,2,\dots  \frac{|\buyers|}{\cost_1}\}$:
    \State \quad\quad  $\itemprice_N \leftarrow A(\{\valj\}_{j\in\buyers}, N, \eps)$.
        \Comment{Execute $A$ to obtain an $|\buyers|\bigO(\eps)$-optimal pricing curve.}
    \State \quad\quad $\mmNj \leftarrow \mjp(N, \itemprice_N)$ for all $j\in\buyers$.
                \Comment{Compute buyer dataset sizes. See~\eqref{eqn:buyerpurchasemodel}.}
    \State \quad\quad ${\rm profit}_N \leftarrow \sum_{i\in\buyers} \itemprice_N(\mmNi) - \cost_1 N$.
                \Comment{Compute the profit if collecting $N$ data.}
    \State $\numdataA \leftarrow \argmax_{N\in \left\{1,2,\dots  \frac{|\buyers|}{\numdata}\right\}} {\rm profit}_N$.
    \Comment{Find the optimal data collection amount.}
    \State $\mmjA \leftarrow \mm_{\numdataA, j}$, \quad $\pricejA \leftarrow \itemprice_{\numdataA}(\mmjA)$.
    \Comment{Optimal prices and dataset sizes for buyer.}
    \State \textbf{Return: } $\numdataA$, $\{\mmjA\}_{j\in\buyers}$, $\{\pricejA\}_{j\in\buyers}$
\end{algorithmic}
\end{algorithm}

}

\newcommand{\insertAlgoMain}{%
    \begin{algorithm}[H]
 
    \caption{\label{alg:mechanism}}
   \begin{algorithmic}[1]
        \State
        {\bfseries Input:} Buyer valuations $\{\valj\}_{j\in\buyers}$, contributor costs $\{\costi\}_{i\in\contributors}$, algorithm for ordered item pricing $A$, approximation
            parameter $\eps$. %an $\eps$-revenue optimal algorithm $A$.
        \State
        \textbf{The broker: }
        \State \hspace{0.6cm} $\numdataA$, $\{\mmjA\}_{j\in\buyers}$, $\{\pricejA\}_{j\in\buyers}$ \leftarrow \text{Execute Algorithm~\ref{alg:oipalgo} with $\{\valj\}_{j\in\buyers}$, $\{\costi\}_{i\in\contributors}$, $A$, $\eps$.  }$
        \State \hspace{0.6cm} Instruct contributors to collect $\{\nnopti\}_{i\in\contributors}$ where,
        \[
                \nnopt_1 = \numdataA-1, \quad \nnopt_2 = 1, \quad \nnopt_j = 0 \text{ for all } j\in \{3,\dots,|\contributors|\}.
        \]
      
        \State
        \textbf{Each contributor $i \in \contributors$:}
        \State
        \hspace{0.6cm}
         Chooses a strategy $\strati = (\colamnti, \subfunci)$. 
        \State 
        \hspace{0.6cm}
         Collects $\nni$ data points $\initdatai$, and submits $\subdatai = \subfunci(\initdatai)$ . 
        \State
        \textbf{The broker:}
        \State 
        \hspace{0.6cm}   Receive  $\{\subdatai\}_{i\in\contributors}$ from contributors.
        \State \hspace{0.6cm} Let   $Y \leftarrow \subdata_1 \cup \subdata_2$ be the combined dataset of the first two contributors.
       % \Comment{The combined dataset}
        \State
        \hspace{0.6cm} Allocate  $\mmoptj$ ($=\mmjA$) randomly chosen points from $Y$ to each buyer $j \in \buyers$.
        \State \hspace{0.6cm} Charge $\priceoptj$ from each buyer $j\in\buyers$ (see~\eqref{eq:buyer_pay}).

        \State
        \hspace{0.6cm} Pay $\payopti$ to each contributor $i \in \contributors$ (see~\eqref{eq:contrib_pay}).

   \end{algorithmic}   
\end{algorithm}
}

%% file: section/1abstract.tex
% \vspace{-0.1in}

\begin{abstract}
We study a data marketplace where a broker intermediates between buyers, who seek to estimate the mean \(\mu\) of an unknown normal distribution \(\Ncal(\mu, \sigma^2)\), and contributors, who can collect data from this distribution at a cost. 
The broker delegates data collection work to contributors, aggregates reported datasets, sells it to buyers, and redistributes revenue as payments to contributors.
% \rcomment{We study a data marketplace where a broker facilitates transactions between buyers and contributors. Buyers seek to estimate the mean \(\mu\) of an unknown normal distribution \(\mathcal{N}(\mu, \sigma^2)\), but have varying valuations based on their estimation error. Contributors, each with different data collection costs, gather samples from this distribution and report them (not necessarily truthfully) to the broker. The broker then sells subsets of the combined dataset to buyers at varying prices and redistributes the revenue to contributors.} 
We aim to maximize welfare or profit under key constraints: individual rationality for buyers and contributors,  incentive compatibility (contributors are incentivized to comply with data collection instructions and truthfully report the collected data), and budget balance (total contributor payments equals total revenue).
We first compute welfare/profit-optimal prices under truthful reporting; however,
to incentivize data collection and truthful data reporting, we adjust them based on discrepancies in contributors' reported data.
This yields a Nash equilibrium (NE) where the two lowest-cost contributors collect all data. We complement this with two impossibility results: \emph{(i)} no nontrivial dominant-strategy incentive-compatible mechanism exists in this problem, and \emph{(ii)} no mechanism outperforms ours in a NE.
% \rcomment{Shall we change the first part of abstract into the original version}
\end{abstract}

%% file: section/2introduction.tex
% \vspace{-0.1in}
\section{INTRODUCTION}
% \vspace{-0.1in}

With the ubiquity of AI, data has become a critical economic resource, driving operations and innovation across many domains.
However, not everyone has the means to collect data on their own
 and thus must rely on \emph{data marketplaces} to acquire the data they need.
In recent years, data marketplaces have emerged across various domains,
including 
materials science~\cite{citrine},
advertising~\cite{googleads},
computer systems~\cite{azuredatashare},
insurance~\cite{datarade},
freight~\cite{freightdat},
logistics~\cite{bloombergeap}, 
and others~\citep{awsdataexchange,ibmdatafabric,snowflake}. 
% As shown in~\fig~\ref{fig:datamarketplace},
% these platforms act as intermediaries,
% facilitating data purchases between buyers and contributors.

%\parahead{Strategic contributors} 
In data marketplaces, strategic data contributors aim to minimize their data collection costs while maximizing compensation from buyers.  
This can lead to behaviors such as under-collecting data to reduce costs, or misreporting their collected datasets to increase payments.  
Such actions undermine trust in buyers, who seek high-quality data for their learning tasks.  
These considerations motivate the central goal of this paper: design data marketplaces that \emph{disincentivize strategic contributor behavior}.

% In data marketplaces, 
% strategic contributors seek to minimize their data collection costs, while maximizing their compensation from buyers.
% This can give rise to strategic contributor behaviors such as under-collecting data to minimize costs, or misreporting their costs or the
% collected dataset to increase their payments.
% Such behavior can undermine trust among buyers who  do not wish to overpay for data,
% and seek high-quality data for their learning tasks.
% This motivates the focus of this paper: to design data marketplaces which disincentivize such strategic \emph{contributor} behavior.

% misreporting their costs, under-collecting data to minimize costs, and fabricating data (or other types of untruthful reporting) to increase their payments.
% This can undermine trust among buyers who  do not wish to overpay for data, and may question the reliability of the data they receive.

%\parahead{Positioning}
Most prior work on data marketplaces assume that data contributors will
always truthfully report data, and focus on
setting prices for buyers and payments to contributors based on the \emph{quantity of data}.
This overlooks the fact that buyers ultimately care about performance on a given learning task, and pricing based on quantity alone is meaningless when strategic contributors can report data untruthfully.
For instance, naively paying contributors based on the size of the dataset they contribute, incentivizes them to report large fabricated (fake) datasets to dishonestly inflate earnings.
Prior work on incentivizing truthful data reporting do not consider broader market constraints (e.g. contributor payments must originate from buyers who derive value from the data they purchase).

% \begin{figure*}[h]
% %\vspace{-0.06in}
% \centering
% \includegraphics[width=0.99\linewidth]{figs/data_marketplace_v2}
% \vspace{-0.12in}
% \caption{\small%
% An illustration of a data marketplace.
% (1) Data contributors collect and submit their data.
% (2) A \emph{broker} operating the marketplace evaluates
%  the quality of their submissions and prices the data.
% (3) The broker sells subsets of this dataset to the buyers.
% (4) The revenue is redistributed back to the contributors.
% \label{fig:datamarketplace}
% }
% \end{figure*}

\subsection{Summary of Contributions}
\label{sec:contributions}

In this work, we study these fundamental challenges in a normal mean estimation setting.  

\paraheadwnogap{Problem Formalism} 
In \S\ref{sec:setup}, we formalize the problem set up. 
A set of contributors~\(\contributors\) act as the \emph{agents}, each able to collect i.i.d.\ samples from an \emph{unknown} normal distribution \(\Ncal(\mu, \sigma^2)\), each incurring a per-sample cost~\(\costi\). 
We assume these costs are public information. A set of buyers~\(\buyers\) seek to estimate~\(\mu\) by purchasing data from the contributors.  
A broker serves as the \emph{principal}, who facilitates these transactions by delegating data collection work to contributors (which may not be followed), and receiving data submissions (which may be reported untruthfully).
The broker then sells subsets of the received data to buyers and redistributes the revenue to contributors.

% \rcomment{A broker designs a mechanism specifying:  
% (a) how much data each contributor should collect,  
% (b) how much data is sold to each buyer,  
% (c) the prices charged to buyers, and  
% (d) the payments distributed to contributors.  
% Each contributor $i$'s strategy determines how much data she collects (which may differ from the broker’s specification), and what she submits. For instance, a contributor may only collect a small amount of data, and fabricate the rest to mislead the broker, say by fitting a distribution to the small dataset she collected, and then sampling many points from it.} 

\subparahead{Mechanism design objectives} We consider two objectives: welfare and contributors’ profit. Welfare is defined as the total value that buyers derive from the allocated data minus contributors’ total data-collection costs. Revenue is the total payment collected from buyers. Contributors’ profit is then defined as revenue minus contributors’ total data-collection costs. The broker itself is not profit-seeking and only serves as the mechanism operator, its objective is to maximize either welfare or profit, while satisfying key market constraints:
\emph{(i) Budget balance (BB):} total contributor payments equal total revenue from buyers.  
\emph{(ii) Individually rational for contributors (IRC):} contributor payments cover data collection costs.  
\emph{(iii) Incentive-compatible for contributors (ICC):}
contributors are incentivized to  behave well, \ie  follow the broker's instructions on data collection, and report the collected data truthfully;
specifically, being well-behaved constitutes a Nash equilibrium (NE) for all contributors.
% contributors should be incentivized to report their costs truthfully, collect the specified amount of data, and submit it truthfully;
% specifically, being well-behaved is the best strategy for a contributor, when others are doing the same.
\emph{(iv) Individually rational for buyers (IRB):} buyers' prices should not exceed their valuation.

\parahead{Impossibility Results}
In~\S\ref{sec:hardness}, we establish two key impossibility results.
First, there exists no nontrivial dominant-strategy incentive-compatible (DSIC) mechanism for this problem,
i.e. collecting the specified amount and reporting it truthfully is the best strategy regardless of others' strategies.
Second,  in any NE of any mechanism, the maximum achievable profit is  $\OPT - (\cost_2 - \cost_1)$ where $\cost_1,\cost_2$ are the costs of the cheapest and second cheapest agents, and
$\OPT$ is a welfare/profit-optimal baseline in which contributors are not strategic.
Both results stem from a key insight: when others collect no data, an agent can fabricate data at no cost from an any normal distribution.

\parahead{Mechanism Design}
Next, we design mechanisms that satisfy the four requirements outlined above while maximizing either the welfare (\S\ref{sec:mechanism_cost_known})
or profit (\S\ref{sec:profit_max}, Appendix~\ref{app:profit_maximization}).

\subparahead{Key algorithmic insights and analysis}
To minimize costs, one might expect the broker to delegate all data collection to the contributor with the lowest cost.  
However, this approach is not ICC, as a sole contributor could fabricate data, and the broker---having no other data source---would be unable to detect it.  
We show, however, that two agents suffice, as  
 the broker can assess each agent's submission by comparing it against the other's.
Thus, we assign data collection to the two cheapest contributors, allocating nearly all points to the cheapest, and just one point to the second-cheapest, whose primary role is to enable the broker to verify the submission of the cheapest.

% Hence, to minimize costs, the mechanism assigns data collection to the two cheapest contributors.  
% The cheapest contributor still collects nearly all the data, while the second-cheapest collects only a minimal amount (just one point for mean estimation), serving primarily to help the broker verify the reliability of the cheapest contributor’s submissions.

Our mechanisms first determine the total data collection amount, buyers' dataset sizes and prices, and contributor payments to maximize welfare/profit, assuming well-behaved contributors.  
However, buyers' actual prices include an additional term based on the difference between the means of the two contributors' reported datasets, and this variation is reflected in the contributor payments.  
This design ensures ICC, as each contributor is incentivized to report truthfully when the other does, in order to minimize the discrepancy.  
This mechanism achieves a welfare (profit) which matches (nearly matches) the 
$\OPT - (\cost_2 - \cost_1)$ upper bound proved in our impossibility result in~\S\ref{sec:hardness}.

\vspace{-5pt}

% Our proofs build on minimax lower bound techniques for normal mean estimation.  
% Using these approaches we show that when other contributors collect and contribute data truthfully, it is best for any contributor to contribute her data truthfully regardless of how much she has collected. Furthermore, by carefully designing contributor payments, we show that the optimal amount of data for contributors to collect is exactly the amount recommended by the mechanism.

% IRB: as buyers will pay less if there are significant discrepancies contributors' submissions, which may indicate unreliable data,
% and
% BB (with probability 1): as the market must be feasible for every random realization of data.

%\emph{Analysis techniques.}

\subsection{Related Work}
%Due to space constraints, we discuss related work in Appendix~\ref{sec:relatedwork}.
\label{sec:relatedwork}

\vspace{-5pt}

\parahead{Incentivizing Data Collection}
Prior work has studied methods to incentivize data collection, both
via payments in principal-agent models ~\citep{cai2015optimum,cummings2015accuracy,fallah2024optimal}, and via reciprocal data sharing in collaborative learning settings~\citep{karimireddy2022mechanisms,huang2023evaluating,sim2020collaborative,blum2021one}.
%Other previous work use multiple agents to estimate target statistics \cite{fallah2024optimal,karimireddy2022mechanisms}.
However, they assume agents will always truthfully report the collected data. %participation and do not focus on misreporting.
Indeed,
in many cases, they are able to design DSIC mechanisms for contributors, which, as we demonstrate, is impossible in our problem.
Moreover, some methods assume specific forms for the valuation of the data for the principal. In contrast, our approach is more general: we only assume that buyers' valuation decreases with their estimation error.

% \subparahead{Principal-agent models for data collection}
% Some studies have explored principal-agent models in data collection, where a principal incentivizes agents to collect data through payments or other means when data collection is costly~\citep{cai2015optimum,huang2023evaluating,cummings2015accuracy}.
% These methods assume specific forms for the valuation of the data for the principal. our approach is more general: we only assume that buyers' valuation decreases with their estimation error.
% Moreover, in these works, and the works on \emph{Data pricing} it is assumed that  contributors will report data truthfully. Indeed, in many cases, they are able to design DSIC mechanisms for contributors, which, as we demonstrate, is impossible in our problem.

\parahead{Incentivizing Truthful Data Reporting}
Previous work on incentivizing truthful reporting of already collected data~\citep{zheng2024truthful,chen2020truthful,ghosh2014buying} do not study many of the market constraints we study here, such as efficiency, IRB, and IRC. In some of these models, agents do not incur costs to collect data.
Moreover, they assume that the  principal has access to a large budget for incentivization, but do not address the source of these funds. In contrast, our setting requires that payments originate from buyers who derive value from the data contributed by contributors.
Other related work~\cite{cummings2015accuracy} assumes that agents submit unbiased data and only add noise under a variance constraint; however, this restriction on agents' strategic behavior is too strong for practical settings.
\citet{chen2018optimal} investigate truthful mechanisms for obtaining data from strategic agents who have different costs; however, these agents can only misreport their costs and not the actual data. \citet{miller2005eliciting} introduced peer prediction for eliciting truthful reports without ground truth, focusing purely on report incentives; our mechanism additionally ensures market feasibility, accounting for contributors’ data collection costs and other key market constraints.

\subparahead{Truthful Contributions in Collaborative Mean Estimation}
Recent work has explored collaborative mean estimation, where a group of strategic agents collaborate to estimate the mean of a  distribution~\citep{chen2023mechanism,clinton2024collaborative,dorner2023incentivizing}. A common technique in our method and these works is comparing an agent’s reported data mean against the mean of other agents to incentivize truthful reporting. However, our analysis techniques are different from these methods, in part because we must adhere to market constraints that do not arise in payment-less collaborative learning settings, and in part because all three methods assume specific forms for the agents' valuation of data.

\subparahead{Moral Hazard}
The issue of contributors submitting fabricated data can also be analyzed under the framework of \emph{moral hazard}, where agents (contributors) deviate from agreed-upon tasks (data collection) with a principal (broker)~\citep{laffont2009theory,holmstrom1979moral,mirrlees1999theory,helpman1975moral}. As in classical moral hazard settings, the outcome of a contributor’s work (the collected dataset) is stochastic. However, unlike typical models, our broker lacks knowledge of the signal distribution given contributors' effort (the learning problem would be trivial if the data distribution were known) and cannot directly observe the outcomes of these efforts, relying instead on potentially dishonest reports from the contributors themselves.

\parahead{Data Pricing}
Our approach to profit maximization leverages an extensive body of work on data pricing~\citep{agarwal2020towards,agarwal2019marketplace,bergemann2019markets,chen2023equilibrium,bergemann2018design,babaioff2012optimal,mehta2021sell,pei2020survey}.
% In recent years, a significant body of work has focused on designing markets and auctions for data. Some studies assume settings where buyers either purchase (or are allocated) the entire dataset or none at all~\citep{agarwal2020towards,agarwal2019marketplace,bergemann2019markets,chen2023equilibrium}. However, this assumption does not align with real-world markets (e.g.~\citep{snowflake,awsdataexchange, citrine,bloombergeap}), where buyers can purchase smaller subsets of data. This limitation also leaves untapped revenue from buyers willing to pay smaller amounts for small datasets. Other works develop approaches for pricing \emph{information}~\citep{bergemann2018design,babaioff2012optimal,mehta2021sell,pei2020survey}, but these approaches do not reflect real-world marketplaces, where brokers sell datasets directly from various sources.
%\subparahead{Pricing Ordered Items}
Of these, pertinent to us, is a line of work on \emph{pricing ordered items}, 
% General multi-item pricing is a notoriously difficult problem, and several works have developed faster algorithms under structural assumptions on items. One pertinent line of work concerns pricing ordered items,
where unit-demand buyers share a common preference ranking over goods~\citep{chawla2022pricing,hartline2005near,chen2024learning}. 
%Recently,~\citet{chen2024learning} developed methods for pricing ordered items in the context of data markets, assuming additional properties such as smoothness.
We will use algorithms from these works when studying envy-free profit maximization.

% \rcomment{lack of Limitation, which is required in checklist , cost unkown}
% In ~\S\ref{sec:mechanism_cost_unknown}, we relax the assumption from~\S\ref{sec:mechanism_cost_known}. Here, the broker does not know the cost of data collection and therefore must design a mechanism that incentivizes contributors to truthfully report both their costs and collected data. To achieve this, we design a mechanism inspired by the Vickrey-Clarke-Groves (VCG) framework. \textcolor{blue}{Can we move part of the abstract here?}

% \kkcomment{Why VCG is insufficient}

%% file: section/3setup.tex
\section{PROBLEM SET UP}
\label{sec:setup}

We now present the problem set up.
We describe the marketplace environment in~\S\ref{sec:env}, followed by outlining the mechanism design problem in~\S\ref{sec:mechdesignproblem}. A table of notations is provided in Appendix~\ref{sec:notation}.

\subsection{Description of the Marketplace}
\label{sec:env}

% We will now describe our problem set up.
% A data marketplace consists of three key actors:  
% \emph{(i)} a finite set of data \emph{contributors} $\contributors = \{1,\dots,|\contributors|\}$,  
% \emph{(ii)} a finite set of \emph{buyers} $\buyers$, and  
% \emph{(iii)} a \emph{trusted broker} who facilitates transactions between contributors and buyers.
% Contributors collect samples from the same but \emph{unknown} normal distribution $\Ncal(\mu, \sigma^2)$.

%We now describe our problem setup.
A data marketplace involves three key actors:  
\emph{(i)} a finite set of \emph{strategic} data contributors $\contributors = \{1,\dots,|\contributors|\}$,  
\emph{(ii)} a finite set of \emph{non-strategic} buyers $\buyers$, and  
\emph{(iii)} a trusted broker who facilitates transactions between contributors and buyers.
Contributors collect samples from a \emph{common} normal distribution $\Ncal(\mu, \sigma^2)$.
Here, $\sigma^2$ is known. Neither the contributors, buyers, or the broker know $\mu$, nor do they have any ancillary information, such as a prior.
Buyers wish to estimate $\mu$ by purchasing data from the contributors, but have different valuations based on estimation errors.
Contributor $i$ incurs a cost $\costi$ to collect each sample, where $\costi$ is \emph{known to the broker}.

% Buyers aim to estimate $\mu$ by purchasing data from contributors, with valuation depending on their estimation error.
% Each contributor~$i$ incurs a private cost~$\costi$ per sample.

\parahead{Interaction Protocol}
The marketplace operates in the following sequence:
% Contributors first report their costs (not necessarily truthfully) to the broker.
% The broker then requests each contributor to collect a certain amount of data.
% Contributors may collect data (not necessarily as instructed) and report it to the broker (not necessarily truthfully).
% The broker sells subsets of this combined dataset to buyers, and redistributes the revenue among contributors.
% These interactions proceed in the following sequence:
\vspace{-0.1in}
\begin{enumerate} [leftmargin=0.2in,itemsep=-0.02in]

    \item \emph{Mechanism design (broker)}:
    \label{itm:mechanismdesign}
    The broker designs and publishes a mechanism $M$ to delegate data collection work to contributors, elicit the collected data, set prices, sell data to buyers, and redistribute revenue back to contributors.

    \item \emph{Delegation of work (broker)}:
    The broker requests each contributor $i$ to collect $\reqamntri = \reqamnti(\{\costi\}_{i\in\contributors}) \in \NN$ samples, where $\reqamnti$ maps the costs to requested quantities.

    \item \emph{Data collection (contributors)}:
    \label{itm:datacollection}
    Each contributor $i$ collects $\colamntri = \colamnti(\reqamntri, \costi) \in \NN$ i.i.d.\ samples at cost $\cost_i \colamntri$, yielding a dataset $\initdatai = \{z_{i,1}, \dots, z_{i,\colamntri}\} \in \RR^{\colamntri}$.
    The function $\colamnti$ maps the broker's request and the contributor’s cost to the actual number of samples  she will collect (which may not be the amount requested).

    \item \emph{Data submission (contributors)}:
    \label{itm:datasubmission}
    Each contributor submits a dataset $\subdatai = \{x'_{i,1}, x'_{i,2}, \dots\} = \subfunci(\initdatai, \reqamntri, \costi) \in \bigcup_{\ell=0}^\infty \RR^\ell$.
    The submission function $\subfunci$ may modify the collected dataset $\initdatai$, enabling strategic behaviors (\eg fabrication) for personal gain.

    \item \emph{Data allocation and purchases (broker)}:
    \label{itm:pricing}
    The broker allocates to each buyer $j \in \buyers$ a random subset of size $\mmrj = \mmj(\{\subdatai\}_{i\in\contributors}, \{\costi\}_{i\in\contributors}) \in \NN$ drawn from $\bigcup_{i \in \contributors} \subdatai$, and charges a price $\pricerj = \pricej(\{\subdatai\}_{i\in\contributors}, \{\costi\}_{i\in\contributors}) \in \RR$.
    % Let $\receiveddata_j = \recdata_j(\{\subdatai\}_{i\in\contributors}, \mmrj)$ be the dataset received by buyer $j$, where $\recdata_j$ is the randomized function which samples a random subset of size $\mmrj$ from the union of all contributors' datasets $\bigcup_{i\in\contributors}\subdatai$.
    Let $\receiveddata_j$ be the dataset received by buyer $j$, obtained by sampling a random subset of size $\mmrj$ from the union of all contributors' datasets $\bigcup_{i\in\contributors}\subdatai$.

    \item \emph{Revenue redistribution (broker)}:
    \label{itm:revenueredistribution}
    The broker redistributes total revenue $\sum_{j \in \buyers} \pricerj$ to contributors via payments $\{\payri\}_{i \in \contributors}$, where $\payri = \payi(\{\subdatai\}_{i\in\contributors}, \{\costi\}_{i\in\contributors}) \in \RR$, and each payment function $\payi$ maps available information to contributor $i$'s payment.

\end{enumerate}

\subparaheadwnogap{Notation}
Primed symbols denote realizations of quantities that are functions of information available to a contributor or broker, \eg in step~\ref{itm:datacollection}, 
the function $\colamnti$ maps the requested amount and cost to the actual amount contributor $i$ will collect, denoted $\colamntri$.

\parahead{Mechanism}
Let $\dataspace = \bigcup_{\ell=0}^\infty \RR^\ell$ denote the space of datasets.
A broker's mechanism is specified by the tuple of (possibly randomized) functions
$M = \left(\{\reqamnti\}_{i\in\contributors}, \{\mmj\}_{j \in \buyers}, \{\pricej\}_{j \in \buyers}, \{\payi\}_{i \in \contributors}\right)$.
Here, $\reqamnti: \RR_+^{|\contributors|} \rightarrow \NN$ determines the amount of data to request from contributor $i$ based on their costs. 
Next, $\mmj: \dataspace^{|\contributors|} \times \RR_+^{|\contributors|} \rightarrow \NN$ determines the size of the dataset sold to buyer~$j$;
it should satisfy $\mmj(\{\subdatai\}_{i\in\contributors}, \cdot)\leq\sum_{i\in\contributors}|\subdatai|$ for all $j\in\buyers$, as the broker cannot sell more data to any buyer than she has received from contributors.
Next,
\,$\pricej : \dataspace^{|\contributors|}  \times \RR_+^{|\contributors|} \to \RR$ specifies the price charged to buyer~$j$.
Finally,  
$\payi : \dataspace^{|\contributors|} \times \RR_+^{|\contributors|} \to \RR$ determines the payment to contributor~$i$.
% The broker designs and publishes these functions in advance in step~\ref{itm:mechanismdesign}.

\parahead{Contributors}
Without loss of generality, we assume contributors are ordered by increasing cost: $\cost_1 \leq \cost_2 \leq \dots \leq \cost_{|\contributors|}$.
After the mechanism is published, each contributor $i$ selects a strategy, which is a tuple of (possibly randomized) functions $\strati = ( \colamnti, \subfunci)$, where:  
$\colamnti: \NN  \times \RR_+ \to \NN$ determines how much data to collect, based on the broker's request and the cost; and  
$\subfunci: \dataspace \times \NN \times \RR_+ \to \dataspace$ determines what dataset to submit, possibly modifying or fabricating the collected dataset.
This formalism accommodates a wide range of strategic behavior, including under-collecting data 
to reduce data collection costs (\ie $\colamntri = \colamnti(\reqamntri, \costi) < \reqamntri$), and fabricating data to appear compliant 
(\ie $|\subdatai| = \reqamnti > |\initdatai| = \colamntri$, where $\subdatai = \subfunci(\initdatai, \reqamnti,  \costi)$).

\subparahead{Well-behaved contributor}
Our goal will be to design mechanisms that incentivize contributors to be ``well-behaved'', \ie truthfully report their costs, collect the requested amount of data, and submit it without alteration.
To formalize this, let $\identity$ denote the identity function, which returns its first argument: 
$\identity(x, y, z, \dots) = x$.  
Define the well-behaved strategy for contributor~$i$ as $\stratopti = (\identity, \identity)$.
Let $\stratopt$  ($\stratoptmi$) denote the strategy profile in which all contributors (all except $i$) are well-behaved.

\subparahead{Contributor utility} 
If a contributor samples $\colamntri$ points and receives payment $\payri \in \RR$, her \emph{ex post} utility is $\payri - \costi \colamntri$.
As these quantities depend on the mechanism $M$ and all contributors' strategies $\strat$,
her \emph{ex ante} utility $\utilci$
can be defined as\footnote{%
Making the dependencies on the mechanism and strategies explicit, 
we can write this as
$\utilci(\mech,\strat) \defeq  \inf_{\mu\in\RR} \EE\left[ \payi\big(\{\subfuncj(\initdataj)\}_{j\in\contributors}, \{\costj\}_{j\in\contributors}\big)
                - \costi \colamnti\left(\reqamnti( \{\costj\}_{j\in\contributors}), \costi\right)\right]$.
To reduce notational clutter, we will use the primed notation going forward.
},
\vspace{-0.05in}
\begin{align*}
                \numberthis\label{eqn:utilc}
    \utilci(\mech, \strat) &\defeq \inf_{\mu\in\RR} \EE\left[ \payri - \costi \colamntri\right]. \\[-0.23in]
\end{align*}
Here, the expectation is over stochasticity in data generation, and any randomness in the mechanism and
contributor strategies. 
We consider the infimum (worst-case) value over $\mu$ as $\mu$ is unknown, and a process should yield reliable estimates across all $\Ncal(\mu,\sigma^2) $. To illustrate further, suppose the true mean is $\mu=\mutilde$.
Suppose a contributor chooses not to collect any data, $\colamntri=0$ (incurring zero cost), and reports a large vector of repeated copies of $\mutilde$, i.e., $\subdatai = (\mutilde, \dots, \mutilde)$.
This benefits buyers, as they may obtain arbitrarily accurate estimates of $\mu$ using the sample mean.
As a result, buyers would be willing to pay more, which increases the contributor’s payment without incurring any cost.
However, this strategy is viable only if the contributor knew \emph{a priori} that $\mu = \mutilde$.
Taking the infimum over $\mu$ accounts for the fact that $\mu$ is unknown and makes $\utilci$ well-defined\footnote{An alternative approach to address uncertainty about $\mu$ 
is to place a prior over $\mu$ and take an expectation with respect to that prior.
While we do not study this Bayesian setting, our high-level ideas can be adapted to it.}.

\parahead{Buyers}
Buyer $j$ receives a dataset $\receiveddata_j$ from the mechanism, pays the specified price $\pricerj$ to the broker, and estimates $\mu$ via the sample mean $\muhat(\receiveddata_j)$.
While prior work on data marketplaces has often modeled buyer valuation as a function of dataset size~\citep{zhang2023survey}, such a model is inadequate when contributors may misreport data.  
A more principled approach, rooted in classical statistics, measures performance by the estimation error $|\muhat(\receiveddata_j)-\mu|$~\citep{stein1981estimation}.  
Motivated by this, 
%we model buyer valuation directly as a decreasing function of the estimation error, which better captures the buyer’s ultimate objective. 
%Instead, 
we adopt a more appropriate model where a buyer’s valuation depends on the estimation error $|\muhat(\receiveddata_j) - \mu|$, which directly reflects the buyer’s ultimate objective. We provide additional background on normal mean estimation in Appendix~\ref{app:mean_intro}.

\begingroup
\allowdisplaybreaks
\subparahead{Buyer valuations}
We will find it useful to define three related valuation functions for buyers:
\begin{enumerate}[leftmargin=0.2in,itemsep=-0.015in]
    \vspace{-0.1in}
    \item \emph{Error-based valuation $\valfuncj$:}
    Different buyers may have differing valuations for varying errors, which we model using buyer-specific valuation functions $\valfuncj: \RR_+ \rightarrow [0,1]$.
    Here, $\valfuncj(e)$ is buyer $j$'s \emph{ex-post} value when her estimation error is $e$. We adopt general error-based valuation functions $\valfuncj$ to capture heterogeneity across buyers. In real-world markets, buyers often differ in their sensitivity to estimation errors: some may place high value on small accuracy improvements, while others may tolerate larger errors. Allowing $\valfuncj$ to be buyer-specific better reflects diverse buyer objectives.
    The next two definitions build on $\valfuncj$.

    \item \emph{Valuation under a mechanism and strategy profile $\valmechj$:}
    Recall that buyer~$j$'s dataset~$\receiveddata_j$ is a random subset of size~$\mmrj$ drawn from all contributors' submissions (step~\ref{itm:pricing} in interaction protocol).  
    It depends on the mechanism~$M$, the contributor strategy profile~$\strat$, and the stochastic datasets~$\{\initdatai\}_{i \in \contributors}$ collected by the contributors.  
    Thus, buyer~$j$'s value under mechanism~$M$ and strategy profile~$\strat$ is:
    \vspace{-4pt}
    \begin{align*}
        \numberthis
        \label{eqn:valb}
        \valmechj(M, \strat)
        \defeq 
            \inf_{\mu\in\RR} \EE
                \left[ \valfuncj\left( 
           \left|\,\muhat(\receiveddata_j)
           \,-\, \mu\,\right| \right)
           \right].\\[-0.25in]
    \end{align*}

\vspace{-4pt}
    \item
        \begingroup
    \allowdisplaybreaks
    \emph{Valuation under clean data $\valdataj$:}
    As a benchmark, we also define a size-based valuation $\valdataj$, where $\valdataj(m)$ is buyer $j$'s expected value upon receiving $m$ i.i.d.\ (clean) samples:
    \vspace{-0.05in}
    \begin{align*}
        \numberthis
        \label{eqn:valdataj}
        \valdataj(m) = \inf_{\mu\in\RR} \EE_{X\sim\Ncal(\mu,\sigma^2)^m}\left[
            \valfuncj\left(\left|\,\muhat(X) - \mu\,\right|\right)
        \right].\\[-0.2in]
    \end{align*}
    \vspace{-0.5in}
    \endgroup
\end{enumerate}
\endgroup

    We will assume that $\valfuncj$ is such that $\valdataj(m)$ is increasing in $m$, \ie if a buyer receives more clean data, her value increases in expectation.
As in~\eqref{eqn:utilc}, the infimum in~\eqref{eqn:valb} and~\eqref{eqn:valdataj} ensures $\valmechj$ and $\valdataj$ are well-defined as $\mu$ is unknown.

\subparahead{Buyer utility}
A buyer's utility is her valuation minus her payment.
Defining this is somewhat subtle, as infima and expectations do not commute.
We consider two candidate definitions:
% , given in~\eqref{eqn:utilbone} and~\eqref{eqn:utilbtwo}:
% \begin{center}
% \hspace{-0.2in}
% \begin{tabular}{p{2.9in}p{2.5in}}
%   \begin{equation}
%     \utilbj(M, \strat)
%     % \defeq  \utilbj(\mmj, \pricej, \strat)
%         \defeq
%         \inf_{\mu\in\RR} \EE\left[ \valfuncj\left( 
%         \left|\muhat(\receiveddata_j) - \mu\,\right|
%        % \left|\,\muhat\big(\genstratmj( \{\initdatai\}_{i\in\contributors} \big) - \mu\,\right|
%        \right)
%        - \pricerj
%        \right]
%     \label{eqn:utilbtwo}
%   \end{equation} 
%   &
%     \hspace{-0.4in}
%   \begin{equation}
%       \utilbj(M, \strat) \defeq
%         %\utilbj(\mmj, \pricej, \strat) \defeq
%         \valmechj(\mech, \strat) - \sup_{\mu\in\RR} \EE\left[ \pricerj \right]
%     \hspace{-0.03in}
%     \label{eqn:utilbone}
%   \end{equation}
% \end{tabular}
% \end{center}
% %\ac{using $g$ and $f$?}
\begin{align}
\hspace{-0.07in}
    &\utilbj(M, \strat) \defeq
        %\utilbj(\mmj, \pricej, \strat) \defeq
        \valmechj(\mech, \strat) - \sup_{\mu\in\RR} \EE\left[ \pricerj \right].
    \hspace{-0.03in}
    \label{eqn:utilbone} \\
    &\utilbj(M, \strat)
    % \defeq  \utilbj(\mmj, \pricej, \strat)
        \defeq
        \inf_{\mu\in\RR} \EE\left[ \valfuncj\left( 
        \left|\muhat(\receiveddata_j) - \mu\,\right|
       % \left|\,\muhat\big(\genstratmj( \{\initdatai\}_{i\in\contributors} \big) - \mu\,\right|
       \right)
       - \pricerj
       \right].
    \hspace{-0.03in}
    \label{eqn:utilbtwo}
\end{align}
Here,~\eqref{eqn:utilbone} measures the gap between the buyer’s worst-case expected value~\eqref{eqn:valb} and her worst-case expected price,
whereas~\eqref{eqn:utilbtwo} computes the buyer’s \emph{ex-post} utility $\valfuncj(e) - \pricej$ for realized error $e$, and price $\pricej'$ followed by a worst-case expectation.
Clearly,~\eqref{eqn:utilbone} $\leq$ \eqref{eqn:utilbtwo}, so guarantees established for~\eqref{eqn:utilbone} also hold for~\eqref{eqn:utilbtwo}.
Accordingly, we adopt~\eqref{eqn:utilbone} as our measure of buyer utility.
That said, as we will see, in our mechanism, both expressions will be equal, as the price is independent of $\mu$.

\subparahead{Non-strategic buyers, known buyer valuations}
As we focus on contributor-side challenges, we assume buyers are non-strategic and that the broker knows buyer valuations~$\{\valfuncj\}_{j \in \buyers}$.  
Extensions to unknown valuations---elicited via bids or learned in repeated markets---is left for future work.

% As we focus on contributor-side challenges,
% we will assume that buyers are not strategic and that brokers are aware of 
% buyer valuations $\{\valfuncj\}_{j\in\buyers}$.
% Extending our model to cases where valuations are unknown but can be elicited via bids or learned in repeated markets is left for future work.

\parahead{Public Information}
We assume that the set of contributors $\contributors$, buyers $\buyers$, buyer valuations $\{\valfuncj\}_{j \in \buyers}$, and the variance $\sigma^2$ are public information known to all parties. Contributors’ costs $\{\costi\}_{i\in\contributors}$ are known only to the contributors themselves and the broker. A contributor’s strategy and the broker’s mechanism may depend on this information, but we suppress this dependence for simplicity.

\subsection{The Mechanism Design Problem}
\label{sec:mechdesignproblem}

\paraheadwnogap{Market Constraints}
We wish to design a mechanism 
% $\mechopt = (\{\nnopti\}_{i\in\contributors}, \{\mmoptj\}_{j\in\buyers}, \{\priceoptj\}_{j\in\buyers}, \{\payopti\}_{i\in\contributors})$
in which buyers and contributors do not lose by participation, contributors are incentivized to behave well, and payments to contributors equal the revenue. We now formalize these desiderata:
\vspace{-0.1in}
\begin{enumerate} [leftmargin=0.2in,itemsep=-0.02in]
\label{mech:requirements}

\item \emph{Budget balance (BB):}
$\mechopt$ satisfies budget balance if the total payments to contributors
is equal to the total revenue, i.e., $\sum_{i\in\contributors} \payri = \sum_{j\in\buyers} \priceri$.

%with probability 1, i.e. for \emph{every realization} of the data $\{\initdatak\}_{k\in\contributors}$, randomness in the mechanism, and randomness in contributor strategies, we have $\sum_{i\in\contributors} \payri = \sum_{j\in\buyers} \priceri$.

\item \emph{Individually rational for contributors (IRC):}
$M$ is \emph{ex-ante} individually rational for the contributors if $\utilci(\mechopt, \stratopt) \geq 0$ for all contributors $i\in\contributors$.

\item \emph{Incentive-compatible for contributors (ICC):}
$\mechopt$ is incentive-compatible for contributors if $\stratopt$ is a Nash equilibrium. That is,  for all contributors $i\in\contributors$, and  for every alternative strategy $s_i$ for  $i$, we have $\utilci(\mechopt, \stratopt) \geq \utilci(\mechopt, (s_i, \stratoptmi))$.
(In~\S\ref{sec:hardness}, we show that no DSIC mechanisms are possible).

 %(Note that contributors submit truthfully, i.e. $\subfunck=\identity$, in $\stratopt$).
%\ac{should be $\pi_i(M,s)? and p_j(M,s)$ instead?}

\item \emph{Individually rational for buyers (IRB):}
$\mechopt$ is \emph{ex ante}  individually rational for the buyers if
$\utilbj(\mechopt, \stratopt) \geq 0$ for all buyers $j\in\buyers$.

\end{enumerate}

\parahead{Efficiency}
A mechanism should satisfy some notion of efficiency while meeting the above criteria, where common objectives include welfare and profit maximization.  
Our contributor-side mechanism supports both.  
Here, we focus on welfare maximization, which is simpler to implement on the buyer side.
In~\S\ref{sec:profit_max} and
Appendix~\ref{app:profit_maximization}, we extend these ideas to profit maximization.

First, let us define the welfare  $\welf(\mech, s; \mu)$ of a mechanism $M$ under a strategy profile $\strat$:
\begin{align*}
    & \welf(M, s; \mu) =
        \EE\bigg[ \sum_{j\in\buyers}\valmechj(M, s) - \sum_{i\in\contributors} \costi \colamntri \bigg].
    \numberthis
    \label{eqn:welfare}
\end{align*}
For an arbitrary mechanism, the welfare may depend on reported data, which itself depends on the actual distribution $\Ncal(\mu,\sigma^2)$.
However, we will see that for our mechanism, under $\stratopt$, the welfare $\welf(\mechopt, \stratopt; \mu)$ is independent of $\mu$.

\subparahead{A baseline for welfare maximization}
Our approach is to define a welfare-optimal benchmark $\OPT$ and design a mechanism that satisfies the four desiderata while achieving welfare close to $\OPT$. A natural baseline is the maximum welfare achievable when contributors are non-strategic and follow $\stratopt$. In this case, since buyers receive i.i.d.\ data, their valuations depend only on the data amounts $\{\mmrj\}_{j\in\buyers}$ (see~\eqref{eqn:valdataj}). Additionally, the data allocated to any buyer cannot exceed the total collected: $\mmrj \leq \sum_{i\in\contributors} \colamntri = \sum_{i\in\contributors} \reqamntri$. These observations define the baseline $\OPT$:

\vspace{-10pt} 

\begingroup
\setlength{\abovedisplayskip}{5pt}
\setlength{\belowdisplayskip}{4pt}
\allowdisplaybreaks
{\small \vspace{-10pt} 
\begin{align*}
    \OPT &= \max_{M }
            \EE\left[ \left(\;\sum_{j\in\buyers}\valj(M, \stratopt) \;-\; \sum_{i\in\contributors} \costi \reqamntri  \;\right) 
                \right] \\ &
              \overset{(a)}{=}  \max_{\{\reqamntri\}_i, \{\mmrj\}_j}
                \left(\sum_{j\in\buyers} \valdataj(\mmrj) - \sum_{i\in\contributors} \costi \reqamntri \right)
                \\
    &= \max_{\{\reqamntri\}_{i} } \left(\max_{\{\mmrj\}_{j}}\;
        \sum_{j\in\buyers}\valdataj(\mmrj) \;-\; \sum_{i\in\contributors} \costi \reqamntri  \;\right) 
        % = \max_{\{\reqamntri\}_{i} } \left(
        %     \sum_{j\in\buyers}  \valdataj\left(\sum_{i\in\contributors} \reqamnti\right)
        %     \right)
    \\
    &\overset{(b)}{=} \max_{\{\reqamntri\}_{i} } \left(\; \sum_{j\in\buyers} \valdataj\left(\sum_{i\in\contributors} \reqamntri\right) 
    \;-\; \sum_{i\in\contributors} \costi \reqamntri \; \right)\\&
   \overset{(c)}{=} \;\sum_{j\in\buyers} \valdataj\left(\reqamntoneopt\right) - \costone\reqamntoneopt,
    \numberthis
    \label{eqn:blwelfare}
\end{align*}%
}%
\endgroup%

where $\reqamntoneopt = \argmax_{N' \in \NN} \sum_{j \in \buyers} \valdataj(N') - \costone N'$. Here, (a) follows from~\eqref{eqn:valdataj}, the fact that welfare depends only on the quantities $\{\reqamntri\}_i, \{\mmrj\}_j$, not on payments or prices, and that $\reqamntri$ does not depend on the reported data. As $\valdataj$ is non-decreasing, step (b) sets $\mmrj = \sum_{i \in \contributors} \reqamntri$ to maximize $\valdataj(\mmrj)$ under the constraint $\mmj \leq \sum_{i \in \contributors} \reqamntri$, reflecting that welfare is maximized when all collected data is allocated to all buyers. Step (c) assigns all data collection to contributor 1, the lowest-cost contributor, to minimize total cost.
% This yields the following proposition for the welfare-optimal baseline when contributors are non-strategic.

% \begin{proposition}
%     Let $\mechopt$ be the mechanism where
%     $\reqamntr_1 = \reqamntoneopt := \argmax_{N\in\NN} \sum_{j\in\buyers}\valj(N) - \costone N$,
%     $\reqamntri = \reqamntiopt := 0$ for all $i \in \contributors\backslash\{1\}$,
%     $\mmrj = \mmjopt := \reqamntoneopt$ for all $j\in\buyers$,
%     $\pricerj = \pricejopt := \valdataj(\reqamntoneopt)$,
%     $\payr_1 = \payoneopt = \sum_{j\in\buyers} \valdataj(\reqamntoneopt)$,
%     $\payri = \payiopt = 0$ for all $i \in \contributors\backslash\{1\}$.
%     This mechanism is IRB, IRC, and BF, and achieves the maximum possible welfare of $\OPT$~\eqref{eqn:blwelfare}.
%     \label{prop:welfarebaseline}
% \end{proposition}

% It is worth noting that in Proposition~\ref{prop:welfarebaseline}, alternative choices of payments and prices can also achieve $\OPT$ while satisfying IRB, IRC, and BF. Next, we design mechanisms that additionally satisfy ICC, while achieving welfare close to $\OPT$. We begin with the setting where costs are known in~\S\ref{sec:}, and relax this assumption in~\S\ref{sec:}.

\parahead{Some Remarks}  
We conclude this section with three observations.
\emph{(1) Buyer valuation for untruthful data:} We do not explicitly model the fact that truthful data maximizes buyer valuations. Doing so would require additional assumptions on $\vali$ and an analysis of optimal estimators. Instead, we define a baseline where \emph{non-strategic} contributors follow the rules, ensuring buyers receive truthful data. We aim to design a mechanism that incentivizes strategic contributors to also follow the rules while approximating the baseline.
\emph{(2) Negative prices/payments:} Since data is random, data-dependent prices and payments are also random.  
We allow \emph{ex post} negative prices (broker pays buyers) and payments (contributors pay broker). Although unconventional, this enables strong ICC guarantees. Notably, in our mechanism, the \emph{ex-ante} (i.e. expected) prices and payments are always non-negative.
\emph{(3) On BB:} Payments cannot exceed revenue to ensure market feasibility.
We additionally require equality to ensure that trade benefits are fully distributed to contributors, leaving no unallocated excess revenue.
Unlike the other requirements, BB must hold always and not just in expectation, as the market must be feasible for every realization of data.

%% file: section/hardness.tex
\section{IMPOSSIBILITY RESULTS}
\label{sec:hardness}

We first present two impossibility results for this problem. Theorem~\ref{thm:dsic} states that no nontrivial dominant-strategy incentive-compatible mechanism exists, motivating our focus on NE in the ICC guarantee. Theorem~\ref{thm:neup}  establishes an upper bound on the maximum achievable welfare in any NE of any mechanism that ensures truthful data reporting. 
The proofs of both theorems are given in Appendix~\ref{app:hardnessproofs}.

\begin{restatable}{thm}{DSIC}
    \label{thm:dsic}
    For any mechanism $\mech$, if $\strat=\cbr{(   \colamnti ,\subfunci)}_{i\in\contributors}$ is a dominant strategy profile, then no contributor collects any data, \ie $\forall i\in\contributors$, $\colamntri = \colamnti(\reqamntri, \costi) =0$.
\end{restatable}

\begin{restatable}{thm}{welfareUpper}
    \label{thm:neup}
    Let $\blprofit$ be as defined in~\eqref{eqn:blwelfare}.
    For any mechanism $M$, 
    %that ensures truthful data reporting, 
    if strategy profile $\strat = \left\{(\colamnti,\identity)\right\}_{i \in \contributors}$ is a Nash equilibrium, then $\welf(M, \strat;\mu) \leq \OPT - (\cost_2 - \cost_1)$.
    %Let $M$ be any mechanism satisfying IRB, and let $\stratopt$ be a Nash equilibrium of $M$. Then $\welf(M, \stratopt) \leq \OPT - (\cost_2 - \cost_1)$.
    %\kkcomment{correct?}
\end{restatable}

Next, we will design a mechanism that achieves the welfare upper bound established here, while at the same time satisfying BB, IRC, ICC, and IRB.

%% file: section/4cost_known.tex
\section{METHODS AND ANALYSIS}
\label{sec:mechanism_cost_known}

We have outlined our mechanism in Algorithm~\ref{alg:mechanism_c_known}.
For the purpose of extending this algorithm for profit maximization, we use four sets of input parameters:
$\inputOPT$ is the optimal baseline welfare or profit, 
$\inputreqamnt$ specifies the total amount of data that needs to be collected, $\{\buyersellj\}_{j\in\buyers}$ specifies the amounts of
data to sell to each buyer, and $\{\buyerexppricej\}_{j\in\buyers}$ is the expected price to charge each buyer.
For welfare maximization, we will set  $\inputOPT = \OPT$, $\inputreqamnt = \reqamntoneopt$, where $\OPT$ and $\reqamntoneopt$ are as defined in~\eqref{eqn:blwelfare}, and set $\buyersellj = 
\reqamntoneopt$, $\buyerexppricej = \valdataj (\reqamntoneopt)$ for all $j$.
We will choose these values differently for profit maximization in Appendix~\ref{app:profit_maximization}.
%\rcomment{should here be $v_j^{iid}$ also?}

We will first describe our mechanism and then motivate our design choices.
In line~\ref{lin:cknownrequest}, our mechanism instructs the cheapest contributor to collect $\inputreqamnt-1$ points, the second cheapest contributor to collect $1$ point, and the remaining to collect no points.
It then aggregates all the received datasets and chooses $\mmrj=\buyersellj$  random points to sell to buyer $j$
(for welfare maximization, the broker will sell all the data to all buyers as $\buyersellj = \reqamntoneopt$).
He then charges buyer $j$ $\pricerj$ as shown below
in~\eqref{eq:buyer_pay},
and issues payments $\payri$ to the contributors as shown below in~\eqref{eq:contrib_pay}.

To state prices and payments,
for $i \in \{1,2\}$, let $\subdatami$ be the dataset submitted by the other contributor, $\reqamntr_{-i} = \reqamntoneopt - \reqamntri$ be the amount of data the other contributors should collect, and $d_i = \cost_i (\reqamntri)^2 / \sigma^2$. Recall that the input $\buyerexppricej$ specifies buyer $j$'s expected price. We can now express $\pricerj$ as:

{\footnotesize
\begin{align*} 
    \label{eq:buyer_pay}
    \pricerj  = & 
      \sum_{i \in \{1,2\}} \mathbb{I}\left( \left| \subdata_i \right| = \reqamntri \right) 
    \bigg( 
    \buyerexppricej \frac{\reqamntri}{\inputreqamnt} + \frac{d_i}{|\buyers|} \frac{\sigma^2}{\reqamntrmi} + \frac{d_i}{|\buyers|} \frac{\sigma^2}{ \reqamntri} 
    \bigg) \\ & 
    -\hspace{-0.05in} \sum_{i \in \{1,2\}} \frac{d_i}{|\buyers|} \bigg( \widehat\mu(\subdata_1) - \widehat\mu(\subdata_2) \bigg)^2. \numberthis
\end{align*}
}
Next, we set payments $\payri=0$ for contributors $i\in\{3,\dots,|\contributors|\}$. For contributors $i\in\{1, 2\}$ we have,

{\footnotesize
\begin{align*}
\payri &= \mathbb{I}\left( | \subdata_i | = \reqamntri\right) 
\bigg( (\inputOPT + \cost_1 - \cost_2)\frac{\reqamntri}{\inputreqamnt} 
+ \cost_i \reqamntri 
+ d_i \frac{\sigma^2}{\reqamntr_{-i}} \\[4pt]
&\quad + d_i \frac{\sigma^2}{\reqamntr_i} \bigg)
- d_i \big(\muhat(\subdata_1) - \muhat(\subdata_2)\big)^2 
. \numberthis \label{eq:contrib_pay}
\end{align*}
}

\newcommand{\insertalgotab}{\hspace{0.0cm}}

%\insertAlgoMain
\begin{algorithm}[t] 
    \caption{Mechanism\label{alg:mechanism_c_known}}
    \begin{algorithmic}[1]
        \State
        \textbf{Input: } Target welfare or profit $\inputOPT$, Total amount of data to collect $\inputreqamnt
 \in \NN$, Amounts to sell to buyers $\{\buyersellj\}_{j \in \buyers} \in \NN^{|\buyers|}$, Expected prices to charge buyers
        $\{\buyerexppricej\}_{j \in \buyers} \in \RR^{|\buyers|}$.

        \State \label{lin:cknownrequest}\insertalgotab Instruct contributors to collect $\{\reqamntri\}_{i\in\contributors}$ where,
        \vspace{-0.05in}
        \begin{align*}
                \reqamntr_1 = \inputreqamnt -1, \  \reqamntr_2 = 1, \  \reqamntrj = 0 \text{ for all } j \geq 3. \\[-0.25in]
                % j \in \{3,\dots,|\contributors|\}
        \end{align*}
        \State 
        \insertalgotab   Receive datasets $\{\subdatai\}_{i\in\contributors}$ from contributors.
        \State
        \insertalgotab Set  $\mmrj= \buyersellj$.
        Buyer $j$ receives $\mmrj$ randomly chosen points from $\bigcup_{i\in\contributors}\subdatai$.
        \State \insertalgotab Charge $\pricerj$ from each buyer $j\in\buyers$ (see~\eqref{eq:buyer_pay}).
        \Comment{$\pricerj$ uses $\buyerexppricej$}
        \State
        \insertalgotab Pay $\payri$ to each contributor $i \in \contributors$ (see~\eqref{eq:contrib_pay}).
   \end{algorithmic}   
\end{algorithm}

% \kkcomment{Don't use $W$ for the optimal profit as we are using $W$ for welfare. }

\paraheadwnogap{Design Choices}
As indicated in~\eqref{eqn:blwelfare}, collecting $\reqamntoneopt$ data maximizes welfare.
To minimize total cost, one would expect all the work to be assigned to the cheapest contributor.
However, if only one contributor were to collect all the data, she could fabricate it; the broker---having no other data source---will not be able to detect it. Fortunately, we show that using just one data point from a second contributor is sufficient to incentivize contributor 1 to follow the rules.
Hence, we ask the second cheapest contributor to collect one point, and the rest is collected by contributor 1.

Contributor payments distribute the profit proportional to contributions, with contributor $i$'s \emph{expected payment} given by
$\costi\reqamntri + (\OPT \hspace{-0.02in}+ \cost_1 -\hspace{-0.03in} \cost_2) \reqamntri/\reqamntoneopt$.
Actual payments are adjusted based on the difference between the means of contributors’ datasets, with a larger difference lowering payments.
We prove that each contributor should report truthfully when the other does so, to minimize this difference, thus achieving ICC.
This difference is reflected in buyer prices to ensure BB and IRB.

% \vspace{0.06in}
% \noindent
The following theorem summarizes the key properties of our mechanism. It satisfies all the required constraints while achieving a welfare that matches the upper bound in Theorem~\ref{thm:neup}, demonstrating that the mechanism is essentially unimprovable. The proof is given in Appendix~\ref{app:mainproof}.  

\begin{restatable}{thm}{thmNIC} \label{thm:multi_thm}
Let $M$ be the mechanism in Algorithm~\ref{alg:mechanism_c_known}
executed with inputs
$\inputOPT = \OPT$, $\inputreqamnt = \reqamntoneopt$, 
$\buyersellj = 
\reqamntoneopt$, and $\buyerexppricej = \valdataj(\reqamntoneopt)$ for all $j\in\buyers$, where $\OPT$ and $\reqamntoneopt$ are as defined in~\eqref{eqn:blwelfare}.
Then, $M$ satisfies BB, IRC, ICC, and IRB.
    Moreover, its expected welfare at the well-behaved NE is $\welf(M, \stratopt;\mu) = \OPT - (\cost_2 - \cost_1)$.
\end{restatable}

\vspace{-10pt}

\parahead{Discussion}
We highlight three observations.
\emph{1) Multiple minimum-Cost Contributors:} If multiple contributors share the minimum cost, evenly distributing data collection among them yields a welfare of $\blprofit$ and reduces the variance in the comparison terms of equations~\eqref{eq:buyer_pay} and~\eqref{eq:contrib_pay}. This lowers the variance of prices and payments which may be desirable in practice.
\emph{2) Beyond Mean Estimation:} Our insights likely generalize to learning tasks beyond mean estimation, such as the Gaussian family, where similar discrepancy-based incentives may be constructed with minor modifications. We anticipate the absence of nontrivial DSIC mechanisms and the sufficiency of two contributors for an optimal NE. However, the requirement of a single sample from the second contributor is specific to normal mean estimation as it has a  single learnable parameter. Generally, the number of samples needed from the second contributor may scale with the model's degrees of freedom.
\emph{3) Envy-free pricing:} The pricing in~\eqref{eq:buyer_pay} may seem ``unfair'' as buyers may pay different prices for identical datasets. A natural fairness requirement in this context is envy-freeness (EF)~\citep{varian1973equity,guruswami2005profit}. 
While EF is not our primary focus here for simplicity, ~\S\ref{sec:profit_max} and Appendix~\ref{app:profit_maximization} show that our mechanism can be adapted for EF profit maximization.

\subsection{Proof Sketch of ICC in Theorem~\ref{thm:multi_thm}}
\label{sec:iccsketch}

We now, outline a proof sketch for incentive-compatibility.
To prove ICC, we must show that $\stratopt$ is a Nash equilibrium of $\mechopt$. 
That is, we will show, $\forall i\in\contributors$, 
    $\forall\rbr{\colamnti,\subfunci}$,
    % \ac{unsure if $\subfuncspace$ has been previously defined?}
    \begin{equation*}
        \utilci\rbr{
            \mechopt, \rbr{
                (\colamnti,\subfunci),
               \stratopt_{-i}
            }
        }
        \leq
        \utilci\rbr{
            \mechopt, \rbr{
                 \stratopt_{i},
                 \stratopt_{-i}
            }
        } . 
    \end{equation*}

% \noindent
Note that because we are only asking the two contributors with the lowest collection costs to collect data, we can focus on $\cbr{1,2}\in\contributors$ as the above inequality is immediately satisfied for all other contributors. We now prove the inequality for these contributors in three steps, each with a corresponding lemma.

% \noindent 
\textbf{Step 1.}
Let $(\colamnti,\subfunci)$ be given.
First we show that a contributor can improve her utility by switching their strategy from $(\colamnti,\subfunci)$ to $(\colamnti,\subfunc_a)$, where $\subfunc_a$ is a function that always submits exactly the requested amount of points $\reqamntri$, while maintaining the same sample mean as $\subfunci$. This is true as our mechanism determines a contributors penalty based on the sample mean of their submitted data and whether or not they submitted exactly $\reqamntri$ points. This is stated formally in the lemma below.

\textbf{(\text{Lemma}~\ref{lem:NIC-part1}):}
Fix any $(\colamnti,\subfunci)$. Let $\subfunc_a:\dataspace\rightarrow\RR^{\reqamntri}$ be a function such that $\forall \initdatai\in\dataspace$, $\muhat\rbr{\subfunc_a(\initdatai)}=\muhat\rbr{\subfunci(\initdatai)}$. Then,
\begin{align*}
    \utilci \rbr{M, \rbr{(\colamnti,\subfunci), \stratopt_{-i} }}
    \leq
    \utilci \rbr{M, \rbr{(\colamnti,\subfunc_a), \stratopt_{-i} }} .
\end{align*}

\textbf{Step 2.} 
The second step is showing that a contributor can improve their utility by switching their strategy from $(\colamnti,\subfunc_a)$ to $(\colamnti,\subfunc_b)$, where $\subfunc_b$ is a function that always submits exactly $\reqamntri$ points and whose sample mean agrees with the sample mean of the original dataset she collected $\initdatai$.
This step uses the term in~\eqref{eq:contrib_pay} which penalizes contributors according to the difference in reported means.
Its proof uses techniques from minimax optimality for normal mean estimation~\citep{lehmann2006theory}. This is stated formally in the lemma below.

\textbf{(\text{Lemma}~\ref{lem:NIC-part2}):}
Fix any $\colamnti$. Consider any $\subfunc_a:\dataspace\rightarrow\RR^{\reqamntri}$ and let $\subfunc_b:\dataspace\rightarrow\RR^{\reqamntri}$ be a function such that 
$\forall \initdatai\in\dataspace$, $\muhat\rbr{\subfunc_b(\initdatai)}=\muhat\rbr{\initdatai}$. Then,
\begin{align*}
    \utilci \rbr{M, \rbr{(\colamnti,\subfunc_a), \stratopt_{-i} }}
    % \qquad
    \leq
    \utilci \rbr{M, \rbr{(\colamnti,\subfunc_b), \stratopt_{-i} }} .
\end{align*}

\textbf{Step 3.}
Finally, we show a contributor can improve their utility by switching their strategy from $(\colamnti,\subfunc_b)$ to $(\identity,\subfunc_b)$. We prove that when a contributor is using $\subfunc_b$ their utility is concave in $\colamntri$, which is maximized when $\colamntri = \reqamntri$. We then show the utility under $(\identity,\subfunc_b)$ is the same as under $\stratopt$.

\textbf{(\text{Lemma}~\ref{lem:NIC-part3}):}
Fix any $\colamnti$. Let $\subfunc_b:\dataspace\rightarrow\RR^{\nnopti}$ be a function such that
$\forall \initdatai\in\dataspace$, $\muhat\rbr{\subfunc_b(\initdatai)}=\muhat\rbr{\initdatai}$.
Then,
\vspace{-0.1in}
\begin{align*}
    \utilci \rbr{M, \rbr{(\colamnti,\subfunc_b), \stratopt_{-i}}}
    \leq
    \utilci \rbr{M, \rbr{\stratopt_{i}, \stratopt_{-i}}} .
\end{align*}

These three steps now combine to prove ICC. For any $(\colamnti,\subfunci)$,
    %Let $\subfunc_a:\dataspace\rightarrow\RR^{\nnopti}$ be a function such that $\forall \initdatai\in\dataspace$, $\muhat\rbr{\subfunc_a(\initdatai)}=\muhat\rbr{\subfunci(\initdatai)}$. Let $\subfunc_b:\dataspace\rightarrow\RR^{\nnopti}$ be a function such that $\forall \initdatai\in\dataspace$, $\muhat\rbr{\subfunc_b(\initdatai)}=\muhat\rbr{\initdatai}$.
    we then have that
        \begin{align*}
        \utilci \rbr{M, \rbr{(\colamnti,\subfunci),\stratopt_{-i} }}
        \leq \utilci \rbr{M, \stratopt } .
    \end{align*}

%% file: section/5profit_maximization.tex
\vspace{-5pt}
\section{ENVY-FREE PROFIT MAXIMIZATION}
\label{sec:profit_max}
\vspace{-5pt}
Our mechanism in Algorithm~\ref{alg:mechanism_c_known} naturally extends to other buyer-side objectives. We illustrate this with a profit-maximization use case in Appendix~\ref{app:profit_maximization}.
We summarize the key ideas here. 

%In this section, we will demonstrate how our mechanism can also be applied for envy-free profit maximization.
For this, let us first define the profit of a mechanism $M$ under a strategy profile $\strat$: the total revenue---i.e., the sum of all payments received from buyers---minus the total cost incurred for data collection: $\sum_{j\in\buyers} \pricerj - \sum_{i\in\contributors}\costi\colamntri$.
Hence, the expected profit is given by,
\begin{align*}
   \hspace{-0.2in} \profit(M, s; \mu) = 
    \EE\Bigg[
        \sum_{j\in\buyers} \pricerj  -
            \sum_{i\in\contributors}\costi\colamntri
        \Bigg].
        \numberthis
        \label{eqn:profit}
\end{align*}

\parahead{A Baseline for Envy-free Profit Maximization}
Similar to~\eqref{eqn:blwelfare}, we first establish a baseline for profit maximization when contributors are non-strategic.
Here, the broker seeks to maximize profit subject to two buyer-side incentive constraints: individual rationality for the buyer (IRB) and envy-freeness (EF).
Intuitively, EF requires that no buyer prefers another’s dataset and price over their own~\citep{varian1973equity,guruswami2005profit}.
In this baseline, as in~\eqref{eqn:blwelfare}, the lowest-cost contributor collects all the data. Unfortunately, unlike in~\eqref{eqn:blwelfare}, no closed-form expressions for the realizations of $\{\reqamntri\}_{i \in \contributors}$, $\{\mmrj\}_{j \in \buyers}$, or $\{\pricerj\}_{j \in \buyers}$ are available in the envy-free profit-maximization problem.

However, to apply Algorithm~\ref{alg:mechanism_c_known}, we must specify its four inputs. We can approximate the required quantities algorithmically using the following methods.

\parahead{Algorithms for Envy-free Profit Maximization}
A key technical insight, proved in Lemma~\ref{lem:equal_two_problem}, is that
\emph{envy-free revenue maximization} reduces to the
\emph{pricing ordered items} (POI) problem introduced by~\citet{chawla2007algorithmic}.
Although POI is \textsf{NP}-hard, prior work provides methods for obtaining $\bigO(\epsilon)$-additive approximations~\citep{chawla2007algorithmic,hartline2007profit,chen2024learning}.
Algorithm~\ref{alg:oipalgo} builds on these methods to compute a profit-optimal solution under truthful contributors.
It searches over data collection amounts; for each, it uses POI techniques to compute an approximately optimal revenue, and the corresponding dataset sizes and prices for buyers.
By selecting the data collection amount that maximizes revenue minus cost (assuming only the cheapest contributor collects data), we obtain $\inputreqamnt$.

\parahead{Envy-free Profit Maximization with Strategic Contributors}
Finally, we invoke Algorithm~\ref{alg:mechanism_c_known}, setting
$\inputreqamnt$ as the data collection amount from the above step, and $\inputOPT$,
$\{\buyersellj\}_{j \in \buyers}$, 
$\{\buyerexppricej\}_{j \in \buyers}$ 
as the corresponding profit, dataset sizes, and prices.
This yields a mechanism with profit 
$\OPT - \cost_2 + \cost_1 - \bigO(\epsilon)$,
where $\OPT$ denotes the optimal baseline profit,
while satisfying IRB, IRC, BB, ICC and EF.

%% file: section/6conclusion.tex
\section{CONCLUSION AND DISCUSSION}
\label{sec:conclusion}

We design a data marketplace where a broker mediates transactions between buyers who aim to estimate the mean of an unknown normal distribution, and strategic data contributors, who collect data at a cost. Our mechanism admits a Nash equilibrium (NE) in which the two lowest-cost contributors collect and truthfully report all the data. This equilibrium maximizes welfare or profit subject to key market constraints. Matching impossibility results establish that our mechanism is unimprovable. 

Future directions include extending the framework beyond mean estimation to other learning tasks. A limitation of our current model is that contributors’ costs are assumed to be known. Relaxing this assumption to unknown costs, where contributors may strategically report both costs and data, introduces substantial additional challenges and is left to future work. It would also be interesting to allow buyers to behave strategically, for example by misreporting their valuations to obtain more favorable allocations, which would lead to a different mechanism-design problem beyond our current scope.

% We designed a data marketplace where a broker facilitates transactions between buyers, who seek to estimate the mean of an unknown normal distribution, and strategic data contributors, who collect data from these distributions at a cost. Our mechanism admits a Nash equilibrium (NE) in which the two lowest-cost contributors collect all the data and report it truthfully. In this equilibrium, we maximize either the welfare or profit while satisfying key market constraints.  
% With matching hardness results, we show that our mechanism is unimprovable.
% Interesting directions for future work are to extend our framework to learning tasks beyond mean estimation, and to relax the assumption that the broker knows the costs.

\section*{Acknowledgments}
This work was supported in part by the National Science Foundation Award IIS-2441796.

%% file: Appendices/mean_estimation.tex
\section{REVIEW OF NORMAL MEAN ESTIMATION} \label{app:mean_intro}
We first briefly review normal mean estimation~\citep{stein1981estimation,james1992estimation}. A process\footnote{ Typically, one need not model such a process explicitly: often, the learner receives $\initdata$ as is, or a random subset, keeping \( Y \) i.i.d. In a marketplace, although data is sampled i.i.d., strategic contributors may misreport data to maximize personal gain, in which case $Y$ is no longer an i.i.d.\ dataset from $\Ncal(\mu,\sigma^2)$. In what follows, we will find it useful to explicitly model such a process induced by a mechanism and contributor strategies. } $M$ collects a finite dataset $\initdata=\{z_\ell\}_{\ell}$ of i.i.d.\ points from a normal distribution $\Ncal(\mu,\sigma^2)$, and reports a finite dataset $Y = \{y_\ell\}_\ell$ to a learner. The learner knows the variance $\sigma^2$, but not the mean $\mu$. She estimates $\mu$ via the sample mean $\muhatsm(Y) = \frac{1}{|Y|}\sum_{y \in Y} y$. Conventionally, one studies the \emph{loss} as a function of the estimation error \( |\widehat{\mu}(Y) - \mu| \)~\citep{stein1981estimation}. Here, we instead model the learner's \emph{value} via a decreasing function \( \valfunc: \mathbb{R}_+ \to [0,1] \), where $\valfuncnob(e)$ is her value for achieving error $e$. The learner's valuation $\val$ for a process $M$, which depends on her error-based valuation $\valfunc$, is:
\begin{align}
    \label{eqn:valerr}
    \val(M) = \inf_{\mu\in\RR} \EE\left[ \valfuncnob\left(\left|\,\muhat(Y) - \mu\,\right|\right)
    \right].
    \end{align}
The expectation is over any randomness in $M$ and in generating the original dataset $\initdata$. The infimum over $\mu$ captures the worst-case performance, ensuring reliability across all possible $\Ncal(\mu,\sigma^2)$ since $\mu$ is unknown. To illustrate, let $\mutilde\in\RR$ be arbitrary. A (pathological) process which always reports repeated $\mutilde$ values $(\mutilde, \dots, \mutilde)$ regardless of the collected dataset, achieves maximal value when $\mu=\mutilde$ but fails elsewhere. The infimum ensures $\val(M)$ is well-defined, despite $\mu$ being unknown.

%% file: Appendices/app_cost_known.tex
%%{\Large \textbf{Appendix}}

\vspace{0.1in}

\section{PROOF OF THEOREM~\ref{thm:multi_thm}}
\label{app:mainproof}

\noindent
\textbf{(BB):} We will first show that the mechanism satisfies the budget balance property, that is the total contributor payments is equal to the total revenue.
% }
\begin{proof}
    From the definition of the payment rule in $M$ and a sequence of algebraic manipulations we have
    \begin{align*}
        \sum_{i\in\contributors}\payri
        &=
        \sum_{i\in\contributors}
        \rbr{
            \rbr{
                  (\OPT \hspace{-0.02in}+\hspace{-0.03in} \cost_1 \hspace{-0.03in}-\hspace{-0.03in} \cost_2) \frac{\reqamntri}{\reqamntoneopt}
                +
                \costi\reqamntri
                +
                \frac{d_i\sigma^2}{\reqamntri}
                +
                \frac{d_i\sigma^2}{\reqamntr_{-i}}
            }
            -d_i\rbr{
                \muhat(\initdatai)-\muhat(\initdata_{-i}) 
            }^2
        }
        \\
        &=
        \frac{  (\OPT \hspace{-0.02in}+\hspace{-0.03in} \cost_1 \hspace{-0.03in}-\hspace{-0.03in} \cost_2) }{\reqamntoneopt} \sum_{i\in\contributors} \reqamntri
        +
        \sum_{i\in\contributors} \costi\reqamntri
        +
        \sum_{i\in\contributors}
        \rbr{
            \frac{d_i\sigma^2}{\reqamntr_{-i}}
            +
            \frac{d_i\sigma^2}{\reqamntri}
            -d_i\rbr{
                \muhat(\initdatai)-\muhat(\initdata_{-i}) 
            }^2
        }
        \\
        &=
          (\OPT \hspace{-0.02in}+\hspace{-0.03in} \cost_1 \hspace{-0.03in}-\hspace{-0.03in} \cost_2)  
        +
        \sum_{i\in\contributors} \costi\reqamntri
        +
        \sum_{j\in\buyers}\sum_{i\in\contributors}
        \rbr{
            \frac{d_i\sigma^2}{\abr{\buyers}\reqamntr_{-i}}
            +
            \frac{d_i\sigma^2}{\abr{\buyers}\reqamntri}
            -
            \frac{d_i}{\abr{\buyers}}
            \rbr{
                \muhat(\initdatai)-\muhat(\initdata_{-i}) 
            }^2
        }
        \\
        &=
        \sum_{j\in\buyers} \buyerexppricej
        +
        \sum_{j\in\buyers}\sum_{i\in\contributors}
        \rbr{
            \frac{d_i\sigma^2}{\abr{\buyers}\reqamntr_{-i}}
            +
            \frac{d_i\sigma^2}{\abr{\buyers}\reqamntr_{i}}
            -
            \frac{d_i}{\abr{\buyers}}
            \rbr{
                \muhat(\initdatai)-\muhat(\initdata_{-i}) 
            }^2
        }
        \\
        &=
        \sum_{j\in\buyers}
        \rbr{
            \sum_{i\in\contributors}
            \rbr{
                \frac{\buyerexppricej \reqamntri}{\reqamntoneopt}
                +
                \frac{d_i\sigma^2}{\abr{\buyers}\reqamntr_{-i}}
                +
                \frac{d_i\sigma^2}{\abr{\buyers}\reqamntri}
                -
                \frac{d_i}{\abr{\buyers}}
                \rbr{
                    \muhat(\initdatai)-\muhat(\initdata_{-i}) 
                }^2
            }
        }
        =
        \sum_{j\in\buyers}
        \pricerj.
    \end{align*}
    Therefore, $M$ satisfies the budget balance property.
\end{proof}

% \parahead{
\noindent
\textbf{(IRC):} Next, we will show that under $M$, $\stratopt$ is \emph{ex-ante} individually rational for contributors.
% }
\begin{proof}
    Under the recommended strategy, we have $\colamntri = \reqamntri$ and $\subdatai = \initdatai $, the utility for contributor $i$ is 
    \begin{align*}
        \utilci(M,\stratopt)
        &=
        \inf_{\mu\in\RR}
        \EE\Bigg[ \mathbb{I}\left( \left| \subdata_i \right| = \reqamntri \right)            
            \rbr{
                (\OPT \hspace{-0.02in}+\hspace{-0.03in} \cost_1 \hspace{-0.03in}-\hspace{-0.03in} \cost_2) \frac{\reqamntri}{\reqamntoneopt} + \cost_i \reqamntri  +  \frac{d_i\sigma^2}{\reqamntr_{-i}} +  \frac{d_i\sigma^2}{ \reqamntri} 
            }
            -  d_i\rbr{
                \muhat(\initdatai)-\muhat(\initdata_{-i})
            }^2
            -\costi\reqamntri
        \Bigg]
        \\
        &=
            \rbr{
                (\OPT \hspace{-0.02in}+\hspace{-0.03in} \cost_1 \hspace{-0.03in}-\hspace{-0.03in} \cost_2) \frac{\reqamntri}{\reqamntoneopt} + \cost_i \reqamntri  +  \frac{d_i\sigma^2}{\reqamntr_{-i}} +  \frac{d_i\sigma^2}{ \reqamntri} 
            }
            -\costi\reqamntri
        -d_i \sup_{\mu\in\RR}
        \EE\Bigg[
         \rbr{
            \muhat(\initdatai)-\muhat(\initdata_{-i})
        }^2
        \Bigg]
        \\
        &=
            \rbr{
                (\OPT \hspace{-0.02in}+\hspace{-0.03in} \cost_1 \hspace{-0.03in}-\hspace{-0.03in} \cost_2) \frac{\reqamntri}{\reqamntoneopt} + \cost_i \reqamntri  +  \frac{d_i\sigma^2}{\reqamntr_{-i}} +  \frac{d_i\sigma^2}{ \reqamntri} 
            }
            -\costi\reqamntri
        -\rbr{
            \frac{d_i\sigma^2}{\reqamntr_{-i}} +  \frac{d_i\sigma^2}{ \reqamntri} 
        }
        \\
        &=
        (\OPT \hspace{-0.02in}+\hspace{-0.03in} \cost_1 \hspace{-0.03in}-\hspace{-0.03in} \cost_2) \frac{\reqamntri}{\reqamntoneopt} 
        \geq 0 .
    \end{align*}

\end{proof}

\begin{remark}
Our mechanism, as presented in Algorithm~\ref{alg:mechanism_c_known}, charges each buyer their \emph{ex-ante} value for the data (under truthful reporting), which is why the buyer's \emph{ex-ante} utility is zero. However, it is straightforward to see that by slightly reducing the price and proportionally decreasing contributor payments, one can also satisfy individual rationality for the contributors (IRC) while ensuring that buyers receive strictly positive utility.
\end{remark}

\noindent
\textbf{(ICC):} We defer the proof of the ICC property to Appendix~\ref{app:tech}.
% }

% \parahead{
\noindent
\textbf{(IRB):} Finally, we will show that under $M$, $\stratopt$ is individually rational for buyers.
% }
\begin{proof}
    Under the recommended strategy, the expected payment for buyer $i$ is
    \begin{align*}
        \mathbb{E}\sbr{\pricerj}   & =  \sum_{i \in  \{1,2\} }\mathbb{I}\left( \left| \subdata_i \right| = \reqamntri \right) 
    \bigg( 
    \buyerexppricej \frac{\reqamntri}{\reqamntoneopt} + \frac{d_i}{|\buyers|} \frac{\sigma^2}{\reqamntr_{-i}} + \frac{d_i}{|\buyers|} \frac{\sigma^2}{ \reqamntri} 
    \bigg)   - \sum_{i  \in \{1,2\} } \frac{d_i}{|\buyers|} \mathbb{E}\sbr{\bigg( \muhat(\subdatai) - \muhat(\subdata_{-i}) \bigg)^2}  \\ & = \;    
    \bigg( 
     \buyerexppricej  + \sum_{i \in  \{1,2\}}  \frac{d_i}{|\buyers|} \frac{\sigma^2}{\reqamntr_{-i}} +\sum_{i \in \contributors}  \frac{d_i}{|\buyers|} \frac{\sigma^2}{ \reqamntri} 
    \bigg)   - \sum_{i \in  \{1,2\}} \frac{d_i}{|\buyers|} \rbr{ \frac{\sigma^2}{\reqamntri} + \frac{\sigma^2}{\reqamntr_{-i}} }  \\ & = \; \buyerexppricej . \numberthis \label{eqn:IRB}
    \end{align*}
        
   Therefore, combing with definition of buyer's utility (\ref{eqn:utilbone}), and buyers's value for i.i.d. data (\ref{eqn:valdataj}), buyer $j$'s utility under recommended strategy is

   \begin{align*}
       \utilbj(M, \stratopt) \defeq
        %\utilbj(\mmj, \pricej, \strat) \defeq
        \valmechj(\mech, \stratopt) - \sup_{\mu\in\RR} \EE\left[ \pricerj \right] = \valdataj(\reqamntoneopt)- \sup_{\mu\in\RR} \EE\left[ \pricerj \right] =\valdataj(\reqamntoneopt)- \buyerexppricej  = 0,
   \end{align*}
where the third step is by (\ref{eqn:IRB}) and the last step is by definition $\buyerexppricej =  \valdataj(\reqamntoneopt)$.
\end{proof}

\noindent
\textbf{Efficiency:}
    % \emph{
Finally, we will show that
under $M$, when the contributors follow $\stratopt$, the expected welfare achieves the optimal welfare upper bound in Theorem~\ref{thm:neup}.
\begin{proof}
Recall that contributor 1 collects $\reqamntr_1 = \reqamntoneopt - 1$ points, and contributor 2 collects $\reqamntr_2 = 1$ point.
Under truthful reporting, as contributors have i.i.d data~\eqref{eqn:valdataj},
the expected welfare~\eqref{eqn:welfare} is
\begin{align*}
       \welf(M, \stratopt;\mu) =  \sum_{j\in\buyers} \valdataj\left(\reqamntoneopt\right) - \costone(\reqamntoneopt-1) - \cost_2 = \OPT + \cost_1 - \cost_2,
\end{align*}
the last step is by definition of welfare-optimal benchmark $ \OPT = \sum_{j\in\buyers} \valdataj\left(\reqamntoneopt\right) - \costone\reqamntoneopt$ in (\ref{eqn:blwelfare}).

%\paragraph{Extension to Capacity Constraints.}Our framework can also accommodate contributor-side capacity constraints. In particular, each contributor may have an upper limit on how much data she can collect. In this case, the broker simply needs to respect these limits when assigning data-collection tasks. The mechanism itself remains unchanged; only the feasible set of allocations is restricted accordingly. Under this modification, the same reasoning and proof continue to apply, so our main guarantees still hold.

\end{proof}

%% file: Appendices/app_technical.tex
\section{PROOF OF ICC}\label{app:tech}

% \noindent
% \textbf{(ICC):}

We will now show the ICC property. We recommend that the readers review the proof sketch in~\S\ref{sec:iccsketch}.
% }
%\begin{proof}

We need to show that $\stratopt$ is a Nash equilibrium of $\mechopt$. 
That is, we will show, $\forall i\in\contributors$, 
    $\forall\rbr{\colamnti,\subfunci}$,
    % \ac{unsure if $\subfuncspace$ has been previously defined?}
    \begin{equation*}
        \utilci\rbr{
            \mechopt, \rbr{
                (\colamnti,\subfunci),
                \stratopt_{-i}
            }
        }
        \leq
        \utilci\rbr{
            \mechopt, \rbr{
                \stratopt_{i},
                \stratopt_{-i}
            }
        } . 
    \end{equation*}

This result is proved via key three steps, formalized by the lemmas below.
    Fix any $(\colamnti,\subfunci)$. Let $\subfunc_a:\dataspace\rightarrow\RR^{\reqamntri}$ be a function such that $\forall \initdatai\in\dataspace$, $\muhat\rbr{\subfunc_a(\initdatai)}=\muhat\rbr{\subfunci(\initdatai)}$. Let $\subfunc_b:\dataspace\rightarrow\RR^{\reqamntri}$ be a function such that $\forall \initdatai\in\dataspace$, $\muhat\rbr{\subfunc_b(\initdatai)}=\muhat\rbr{\initdatai}$.
    % Fix any $(\colamnti,\subfunci)$. Let $\subfunc_a:\dataspace\rightarrow\dataspace$ be a function such that $\forall \initdatai\in\dataspace$, $\abr{\subfunc_a(\initdatai)}=\abr{\subfunci(\initdatai)}$ and $\muhat\rbr{\subfunc_a(\initdatai)}=\muhat\rbr{\initdatai}$. Let $\subfunc_b:\dataspace\rightarrow\RR^{\datanum_i^{\star}}$ be a function such that
    % $\forall \initdatai\in\dataspace$, $\muhat\rbr{\subfunc_b(\initdatai)}=\muhat\rbr{\subfunc_a(\initdatai)}$.
    We then have that
    \begin{align*}
        \utilci \rbr{M, \rbr{(\colamnti,\subfunci),\stratopt_{-i} }}
        &\leq
        \utilci \rbr{M, \rbr{(\colamnti,\subfunc_a),\stratopt_{-i} }}
        % \qquad
        &
        (\text{Lemma}~\ref{lem:NIC-part1})
        \\
        &\leq
        \utilci \rbr{M, \rbr{(\colamnti,\subfunc_b),\stratopt_{-i} }}
        % \qquad
        &
        (\text{Lemma}~\ref{lem:NIC-part2})
        \\
        &\leq
        % \utilci \rbr{M, (\datanum_i^{\star},\subfunc_b),\stratopt_{-i} }
        % =
        \utilci \rbr{M, \rbr{\stratopt_{i},\stratopt_{-i} }} .
        &
        % \qquad
        \hspace{-1cm}
        (\text{Lemma}~\ref{lem:NIC-part3})
    \end{align*}
% \end{proof}

\begin{lemma}
   \label{lem:NIC-part1} 
   Fix any $(\colamnti,\subfunci)$. Let $\subfunc_a:\dataspace\rightarrow\RR^{\reqamntri}$ be a function such that $\forall \initdatai\in\dataspace$, $\muhat\rbr{\subfunc_a(\initdatai)}=\muhat\rbr{\subfunci(\initdatai)}$. Then,
\begin{align*}
    \utilci \rbr{M, \rbr{(\colamnti,\subfunci), \stratopt_{-i} }}
    \leq
    \utilci \rbr{M, \rbr{(\colamnti,\subfunc_a), \stratopt_{-i} }} .
\end{align*}
\end{lemma}

\begin{proof}
    Under $M$ we have that 
    \begin{align*}
        &\hspace{0.4cm}
        \utilci \rbr{M, \rbr{(\colamnti,\subfunci),\stratopt_{-i} }}
        \\
        &=
        \inf_{\mu\in\RR}
        \EE\Bigg[
            \mathbb{I}\left( \left| Y_i \right| = \reqamntri \right) 
            \rbr{
                (\OPT \hspace{-0.02in}+\hspace{-0.03in} \cost_1 \hspace{-0.03in}-\hspace{-0.03in} \cost_2) \frac{\reqamntri}{\reqamntoneopt} + \cost_i \reqamntri  +  \frac{d_i\sigma^2}{\reqamntr_{-i}} +  \frac{d_i\sigma^2}{ \reqamntri} 
            }
        % \\
        % &\hspace{3cm}
            -  d_i\rbr{
                \muhat(\subfunci(\initdatai))-\muhat(\initdata_{-i})
            }^2
            -\costi\colamntri
        \Bigg]
        \\
        &\leq
        \inf_{\mu\in\RR}
        \EE\Bigg[
            \rbr{
                (\OPT \hspace{-0.02in}+\hspace{-0.03in} \cost_1 \hspace{-0.03in}-\hspace{-0.03in} \cost_2) \frac{\reqamntri}{\reqamntoneopt} + \cost_i \reqamntri  +  \frac{d_i\sigma^2}{\reqamntr_{-i}} +  \frac{d_i\sigma^2}{ \reqamntri} 
            }
            -  d_i\rbr{
                \muhat(\subfunc_a(\initdatai))-\muhat(\initdata_{-i})
            }^2
            -\costi\colamntri
        \Bigg]
        \\
        &=
        \utilci \rbr{M, \rbr{(\colamnti,\subfunc_a),\stratopt_{-i} }} .
    \end{align*} 
    The inequality follows from the definition of $\subfunc_a$. 
\end{proof}

\begin{lemma}
   \label{lem:NIC-part2} 
   Fix any $\colamnti$. Consider any $\subfunc_a:\dataspace\rightarrow\RR^{\reqamntri}$ and let $\subfunc_b:\dataspace\rightarrow\RR^{\reqamntri}$ be a function such that 
$\forall \initdatai\in\dataspace$, $\muhat\rbr{\subfunc_b(\initdatai)}=\muhat\rbr{\initdatai}$. Then,
\begin{align*}
    \utilci \rbr{M, \rbr{(\colamnti,\subfunc_a), \stratopt_{-i} }}
    % \qquad
    \leq
    \utilci \rbr{M, \rbr{(\colamnti,\subfunc_b), \stratopt_{-i} }} .
\end{align*}
\end{lemma}

\begin{proof}
    From the definition of $\subfunc_a$ we have
    % Since $\subfunc_a$ maps to $\RR^{\reqamntri}$, the indicator term in the payment function is always 1, so we get
    \begin{align*}
        % &\hspace{0.45cm}
        &\utilci \rbr{M, \rbr{(\colamnti,\subfunc_a),\stratopt_{-i} }}
         \\
        =&
        \inf_{\mu\in\RR}
        \EE\Bigg[
            \rbr{
                (\OPT \hspace{-0.02in}+\hspace{-0.03in} \cost_1 \hspace{-0.03in}-\hspace{-0.03in} \cost_2) \frac{\reqamntri}{\reqamntoneopt} + \cost_i \reqamntri  +  \frac{d_i\sigma^2}{\reqamntr_{-i}} +  \frac{d_i\sigma^2}{ \reqamntri} 
            }
            -  d_i\rbr{
                \muhat(\subfunc_a(\initdatai))-\muhat(\initdata_{-i})
            }^2
            -\costi\colamntri
        \Bigg]
        \\
        =&
        \rbr{
            (\OPT \hspace{-0.02in}+\hspace{-0.03in} \cost_1 \hspace{-0.03in}-\hspace{-0.03in} \cost_2) \frac{\reqamntri}{\reqamntoneopt} + \cost_i \reqamntri  +  \frac{d_i\sigma^2}{\reqamntr_{-i}} +  \frac{d_i\sigma^2}{ \reqamntri} 
        }
        -\costi\colamntri
        +d_i
        \inf_{\mu\in\RR}
        \EE\sbr{
            -
            \rbr{
                \muhat(\subfunc_a(\initdatai))-\muhat(\initdata_{-i})
            }^2
        }
        \\
       =&
        \rbr{
            (\OPT \hspace{-0.02in}+\hspace{-0.03in} \cost_1 \hspace{-0.03in}-\hspace{-0.03in} \cost_2) \frac{\reqamntri}{\reqamntoneopt} + \cost_i \reqamntri  +  \frac{d_i\sigma^2}{\reqamntr_{-i}} +  \frac{d_i\sigma^2}{ \reqamntri} 
        }
        -\costi\colamntri
        -d_i
        \sup_{\mu\in\RR}
        \EE\sbr{
            \rbr{
                \muhat(\subfunc_a(\initdatai))-\muhat(\initdata_{-i})
            }^2
        }.
        \numberthis
        \label{eq:b2-largeeq}
    \end{align*} 
    Now notice that because the other agents truthfully submit a combined $\reqamntr_{-i}$ points under $\stratopt_{-i}$, we have that 
    \begin{align*}
        \sup_{\mu\in\RR}
        \EE\sbr{
            \rbr{
                \muhat(\subfunc_a(\initdatai))-\muhat(\initdata_{-i})
            }^2
        }
        &=
        \sup_{\mu\in\RR}
        \rbr{
            \EE\sbr{
                \rbr{
                    \muhat(\subfunc_a(\initdatai))-\mu
                }^2
            }
            +
            \EE\sbr{
                \rbr{
                    \muhat(\initdata_{-i})-\mu
                }^2
            }
        }
        \\
        &=
        \sup_{\mu\in\RR}
        \rbr{
            \EE\sbr{
                \rbr{
                    \muhat(\subfunc_a(\initdatai))-\mu
                }^2
            }
        }
        +\frac{\sigma^2}{\reqamntr_{-i}} .
    \end{align*}
    Therefore, we can rewrite~\eqref{eq:b2-largeeq}, using the fact that the sample mean is minimax optimal~\citep{lehmann2006theory} and the definition of $\subfunc_b$, to get
    \begin{align*}
        &\hspace{0.4cm}
        \rbr{
            (\OPT \hspace{-0.02in}+\hspace{-0.03in} \cost_1 \hspace{-0.03in}-\hspace{-0.03in} \cost_2) \frac{\reqamntri}{\reqamntoneopt} + \cost_i \reqamntri  +  \frac{d_i\sigma^2}{\reqamntr_{-i}} +  \frac{d_i\sigma^2}{ \reqamntri} 
        }
        -\costi\colamntri
        -\frac{d_i\sigma^2}{\reqamntr_{-i}}
        -d_i
        \sup_{\mu\in\RR}
        \EE\sbr{
            \rbr{
                \muhat(\subfunc_a(\initdatai))-\mu
            }^2
        }
        \\
        &\leq
        \rbr{
            (\OPT \hspace{-0.02in}+\hspace{-0.03in} \cost_1 \hspace{-0.03in}-\hspace{-0.03in} \cost_2) \frac{\reqamntri}{\reqamntoneopt} + \cost_i \reqamntri  +  \frac{d_i\sigma^2}{\reqamntr_{-i}} +  \frac{d_i\sigma^2}{ \reqamntri} 
        }
        -\costi\colamntri
        -\frac{d_i\sigma^2}{\reqamntr_{-i}}
        -d_i
        \sup_{\mu\in\RR}
        \EE\sbr{
            \rbr{
                \muhat(\initdatai)-\mu
            }^2
        }
        \\
        &=
        \rbr{
            (\OPT \hspace{-0.02in}+\hspace{-0.03in} \cost_1 \hspace{-0.03in}-\hspace{-0.03in} \cost_2) \frac{\reqamntri}{\reqamntoneopt} + \cost_i \reqamntri  +  \frac{d_i\sigma^2}{\reqamntr_{-i}} +  \frac{d_i\sigma^2}{ \reqamntri} 
        }
        -\costi\colamntri
        -d_i
        \sup_{\mu\in\RR}
        \EE\sbr{
            \rbr{
                \muhat(\initdatai)-\muhat(\initdata_{-i})
            }^2
        }
        \\
        &=
        \utilci \rbr{M, \rbr{(\colamnti,\subfunc_b),\stratopt_{-i} }} .
    \end{align*}
    
\end{proof}
% \ac{change $Y_i$ to $\subfunc_a(X_i)$ and $\subfunc_b(X_i)$ throughout lemmas!!!}

\begin{lemma}
   \label{lem:NIC-part3} 
   Fix any $\colamnti$. Let $\subfunc_b:\dataspace\rightarrow\RR^{\reqamntri}$ be a function such that
$\forall \initdatai\in\dataspace$, $\muhat\rbr{\subfunc_b(\initdatai)}=\muhat\rbr{\initdatai}$.
Then,
\vspace{-0.1in}
\begin{align*}
    \utilci \rbr{M, \rbr{(\colamnti,\subfunc_b), \stratopt_{-i}}}
    \leq
    \utilci \rbr{M, \rbr{\stratopt_{i}, \stratopt_{-i}}} .
\end{align*}
\end{lemma}
\begin{proof}
    From the definition of $\subfunc_b$ we have that 
    \begin{align*}
        &\utilci \rbr{M, \rbr{(\colamnti,\subfunc_b),\stratopt_{-i} }}
        \\= &
        \inf_{\mu\in\RR}
        \EE\Bigg[
            \rbr{
                (\OPT \hspace{-0.02in}+\hspace{-0.03in} \cost_1 \hspace{-0.03in}-\hspace{-0.03in} \cost_2) \frac{\reqamntri}{\reqamntoneopt} + \cost_i \reqamntri  +  \frac{d_i\sigma^2}{\reqamntr_{-i}} +  \frac{d_i\sigma^2}{ \reqamntri} 
            }
            -  d_i\rbr{
                \muhat(\subfunc_b(\initdatai))-\muhat(\initdata_{-i})
            }^2
            -\costi\colamntri
        \Bigg]
        \\
        =&
        \rbr{
            (\OPT \hspace{-0.02in}+\hspace{-0.03in} \cost_1 \hspace{-0.03in}-\hspace{-0.03in} \cost_2) \frac{\reqamntri}{\reqamntoneopt} + \cost_i \reqamntri  +  \frac{d_i\sigma^2}{\reqamntr_{-i}} +  \frac{d_i\sigma^2}{ \reqamntri} 
        }
        -\costi\colamntri
        -d_i
        \sup_{\mu\in\RR}
        \EE\sbr{
            \rbr{
                \muhat(\initdatai)-\muhat(\initdata_{-i})
            }^2
        }
        \\
        =&
        \rbr{
            (\OPT \hspace{-0.02in}+\hspace{-0.03in} \cost_1 \hspace{-0.03in}-\hspace{-0.03in} \cost_2) \frac{\reqamntri}{\reqamntoneopt} + \cost_i \reqamntri  +  \frac{d_i\sigma^2}{\reqamntr_{-i}} +  \frac{d_i\sigma^2}{ \reqamntri} 
        }
        -\costi\colamntri
        -\frac{d_i\sigma^2}{\colamntri}
        -\frac{d_i\sigma^2}{\reqamntr_{-i}}
        .
    \end{align*}
    Define the concave function $l:\RR_+\rightarrow\RR$, given by $l(\colamntri)=-\costi\colamntri-\frac{d_i\sigma^2}{\colamntri}$. Observe that $l'(\colamntri)=-\costi+\frac{d_i\sigma^2}{{\colamntri}^2}$ so $l'(\colamntri)=0\iff \colamntri=\sigma\sqrt{\frac{d_i}{\costi}}=\reqamntri$. Therefore the choice of $\colamntri$ which maximizes $\utilci \rbr{M, \rbr{(\colamnti,\subfunc_b),\stratopt_{-i} }}$ is $\reqamntri$. Together with the fact that $\muhat(\subfunc_b(\initdatai))=\muhat(\initdatai)=\muhat(\identity(\initdatai))$ we conclude
    \begin{align*}
        \utilci \rbr{M, \rbr{(\colamnti,\subfunc_b),\stratopt_{-i} }}
        \leq
        \utilci \rbr{M, \rbr{(\identity,\subfunc_b),\stratopt_{-i} }}
        =
        \utilci \rbr{M, \rbr{(\identity,\identity),\stratopt_{-i} }}
        =
        \utilci \rbr{M, \rbr{\stratopti,\stratopt_{-i} }} .
    \end{align*}
\end{proof}

%% file: Appendices/app_hardness_result.tex
\section{PROOF OF IMPOSSIBILITY RESULT IN SECTION~\S\ref{sec:hardness}}
\label{app:hardnessproofs}

%\subsection{Proofs of Theorems~\ref{thm:dsic} and~\ref{thm:neup}}

We will now prove our impossibility results in Section~\S\ref{sec:hardness}.
In this section, with a slight abuse of notation, we will use $0$ to denote both the scalar zero, and the zero function.

The following lemma is the key technical ingredient in proving both Theorem~\ref{thm:dsic} and Theorem~\ref{thm:neup}.
It states that when other agents are not collecting any data, then, regardless of their submission functions, the best response for an agent is to also not collect any data.
Note that even when an agent has not collected any data, an untruthful submission function may actually report some data.

\begin{lemma}
    \label{lem:otherscollectzero}
        Let $M$ be any mechanism and let
        $\stratmi = \{(0, \subfuncj)\}_{j\neq i}$ be a strategy profile for all agents except $i$, where
        $\{\subfuncj\}_{j\neq i}$ are arbitrary submission functions.
        If $\strati = (  \colamnti, \subfunci)$ is the best response of agent $i$ to $\stratmi$, then $\colamnti=0$ (\ie $   \colamntri = \colamnti(\reqamntri, \costi) =0$ for all $\reqamntri, \costi$).
\end{lemma}
%\rcomment{here we use $\colamntri = 0$ or create a new notation $\mathbb{0}$ to denote the 0 function}
We will prove Lemma~\ref{lem:otherscollectzero} in~\S\ref{sec:otherscollectzero}.
We first prove Theorem~\ref{thm:dsic} and Theorem~\ref{thm:neup}
%Lemma~\ref{lem:nic-under-single-collector}
using Lemma~\ref{lem:otherscollectzero}.

\DSIC* 
%restate of theorem
\begin{proof}[Proof of Theorem~\ref{thm:dsic}]
    Let $\stratmi = \{( 0, \subfuncj)\}_{j\neq i}$ be a strategy profile for all agents except $i$, where $\cbr{\subfuncj}_{j\neq i}$ are arbitrary submission functions. Since $s$ is a dominant strategy profile we know that $\forall \strati'$
    \begin{align*}
        \utilci\rbr{M, \rbr{\strati, \stratmi}}
        \geq 
        \utilci\rbr{M,\rbr{\strati', \stratmi}}
    \end{align*}
    i.e. $\strati$ is the best response to $\stratmi$. By Lemma~\ref{lem:otherscollectzero} we conclude that $  \colamntri=0$.
\end{proof}

\welfareUpper* %restate of theorem

We will first present a technical result, which follows as a corollary of Lemma~\ref{lem:otherscollectzero}. We will then use this result to prove Theorem~\ref{thm:neup}.

\begin{corollary}[Corollary of Lemma~\ref{lem:otherscollectzero}]
    \label{lem:nic-under-single-collector}
    If $\strat=\cbr{(  \colamnti, \subfunci)}_{i\in\contributors}$ is a NE under $\mech$ and $\forall j\neq i$, $\colamnt_j=0$, then $  \colamnti=0$.
\end{corollary}

\begin{proof}
Since $\strat$ is a NE under $\mech$, $\forall \strati'$ we have that 
% \begin{align*}%
    $\utilci\rbr{M, \rbr{\strati, \stratmi}}
    \geq 
    \utilci\rbr{M,\rbr{\strati', \stratmi}}$,
% \end{align*}
i.e. $\strati$ is the best response to $\stratmi$. 
By the hypothesis of the corollary, we also know that $\stratmi=\cbr{(0,\subfuncj)}_{j\neq i}$. Therefore, we can apply Lemma~\ref{lem:otherscollectzero} to conclude $  \colamnti = 0$. 
\end{proof}

\begin{proof}[Proof of Theorem~\ref{thm:neup}]
    As defined in \eqref{eqn:welfare}, the welfare  $\welf(\mech, s;\mu)$ of a mechanism $M$ under a strategy profile $\strat$ is \[\welf(M, s; \mu) =  \EE\bigg[ \sum_{j\in\buyers}\valmechj(M, s) - \sum_{i\in\contributors} \costi \colamntri \bigg].\]
    Under strategy $s$, the contributors are submitting a total of $\numdata  \defeq \sum_{i} 
    \colamntri$ truthful points to the broker. By buyers's value for i.i.d. data \eqref{eqn:valdataj} and \eqref{eqn:valb}, we have $\valmechj(M, s) =  \valdataj(\mmrj)$. Since $\mmrj\leq \numdata$, it follows that $ \valdataj(\mmrj) \leq \valdataj(\numdata)$. Therefore, $$ \welf(M, \strat; \mu) =  \EE\bigg[ \sum_{j \in \buyers}\valdataj(\mmrj) - \sum_{i \in \contributors} \costi \colamntri \bigg] \leq   \EE\bigg[ \sum_{j \in \buyers}\valdataj(\numdata) - \sum_{i \in \contributors} \costi \colamntri \bigg] $$

    Corollary~\ref{lem:nic-under-single-collector}
    tells us that under the Nash equilibrium $s=\cbr{(  \colamnti, \identity)}_{i\in\contributors}$, at least two contributors must collect data. Assuming that $\cost_1 \leq \cost_2 \leq \dots \leq \cost_{|\buyers|}$, we have that the cheapest way for at least two contributors to collect a total of $\numdata  \defeq \sum_{i} 
    \colamntri$ points is to have agent 1 collect $\numdata-1$ points and agent 2 collect 1 point. Combining this with the fact that 
    $\OPT$ is the best welfare achievable when contributors are non-strategic, we get
    \begin{align*}
        ~\welf(M, s;\mu) 
        \leq & \;  
         \EE\bigg[ \sum_{j \in \buyers}\valdatai\rbr{ 
            \numdata 
        } 
        - \sum_{i\in\contributors} \cost_i \colamntri \bigg] \leq \;  
         \EE\bigg[ \sum_{j \in \buyers}\valdatai \rbr{ 
            \numdata 
        } 
        -  \cost_1 \rbr{\numdata -1} -  \cost_2 \bigg] 
        \\ 
        \leq   &  \;   \max_N \cbr{  \sum_{i \in \contributors}\vali 
 \left(N\right) - \costone N} -  (\cost_2 -\cost_1)  
        =   \; \blprofit - (\cost_2 - \cost_1) .
    \end{align*}
 
    The last step follows from equation (\ref{eqn:blwelfare}), the definition of $\blprofit$. 
\end{proof}

 %    Corollary~\ref{lem:nic-under-single-collector}
 %    tells us that under any Nash equilibrium $s=\cbr{(  \colamnti, \subfunci)}_{i\in\contributors}$, at least two contributors must collect data. Assuming that $\cost_1 \leq \cost_2 \leq \dots \leq \cost_{|\buyers|}$, we have that the cheapest way for at least two contributors to collect a total of $\numdata  \defeq \sum_{i} 
 %    \colamntri$ points is to have agent 1 collect $\numdata-1$ points and agent 2 collect 1 point. Combining this with the fact that 
 %    $\OPT$ is the best welfare achievable under IRB, we get
 %    \begin{align*}
 %        ~\welf(M, s) = & \;  
 %        \sum_{i \in \contributors}\vali\rbr{ 
 %            \numdata 
 %        } 
 %        - \sum_{i\in\contributors} \cost_i \colamntri \leq \;  
 %        \sum_{i \in \contributors}\vali \rbr{ 
 %            \numdata 
 %        } 
 %        -  \cost_1 \rbr{\numdata -1} -  \cost_2 
 %        \\ 
 %        \leq   &  \;   \max_N \cbr{  \sum_{i \in \contributors}\vali 
 % \left(N\right) - \costone N} -  (\cost_2 -\cost_1)  
 %        =   \; \blprofit - (\cost_2 - \cost_1) .
 %    \end{align*}
 
 %    The last step follows from equation (\ref{eqn:blwelfare}), the definition of $\blprofit$. 

% \begin{proof}[Proof of Theorem~\ref{thm:neup}]
%     foobar
% \end{proof}

\subsection{Proof of Lemma~\ref{lem:otherscollectzero}}
\label{sec:otherscollectzero}

% \kkcomment{Begin the proof by explaining why this should be  intuitively true.}

At a high level, Lemma~\ref{lem:otherscollectzero} claims that when the other agents do not collect data, it is best for agent $i$ to not collect data. This is because the mechanism has no knowledge of the true value of $\mu\in\RR$ and thus cannot penalize agent $i$ for submitting fabricated data since the others to do not provide data to validate against. Therefore, agent $i$ has the ability to forgo the cost of data collection without decreasing the payment they receive. Hence, agents $i$'s best response involves collecting no data. 
The following proof formalizes this intuition.

\begin{proof}
    Suppose for a contradiction that $  \colamntri = \colamnti(\reqamntri, \costi) >0$. Define the strategy $\strat_{i,\mu'}=\rbr{    \colamnti,\subfunc_{i,\mu'}}$
    where $\subfunc_{i,\mu'}$ is the submission function which disregards $\initdatai$ and instead samples and submits $  \colamntri$ points from $N(\mu',\sigma^2)$ using $\subfunci$, i.e. $\subfunc_{i,\mu'}(\cdot):=\subfunci(Z_i),~Z_i\sim N(\mu',\sigma^2)^{  \colamntri}$.  
    Applying the definition of contributor $i$'s utility we see that
    \begin{align*}
        \inf_{\mu\in\RR}
        \EEV{
            \initdatai\sim N(\mu,\sigma^2)^{  \colamntri}
        }
        \sbr{
            \pay_i\rbr{
                \subfunci(\initdatai), \rbr{\subfunc_j(\varnothing)}_{j\neq i}
            }
        }
        -\costi  \colamntri
        &\geq 
        \inf_{\mu\in\RR}
        \EEV{
            Z_i\sim N(\mu',\sigma^2)^{  \colamntri}
        }
        \sbr{
            \pay_i\rbr{
                \subfunci(Z_i), \rbr{\subfunc_j(\varnothing)}_{j\neq i}
            }
        }
        -\costi  \colamntri
        \\
        \Rightarrow
        \inf_{\mu\in\RR}
        \EEV{
            \initdatai\sim N(\mu,\sigma^2)^{  \colamntri}
        }
        \sbr{
            \pay_i\rbr{
                \subfunci(\initdatai), \rbr{\subfunc_j(\varnothing)}_{j\neq i}
            }
        }
        &\geq 
        \inf_{\mu\in\RR}
        \EEV{
            Z_i\sim N(\mu',\sigma^2)^{  \colamntri}
        }
        \sbr{
            \pay_i\rbr{
                \subfunci(Z_i), \rbr{\subfunc_j(\varnothing)}_{j\neq i}
            }
        } .
    \end{align*}
    For clarity we omit writing the dependence on $M$ and $\strat_{-i}$ under the expectation as it is not relevant to the proof. Now notice that since all of the other agents are not collecting data we have
    \begin{align*}
        \inf_{\mu\in\RR}
        \EEV{
            Z_i\sim N(\mu',\sigma^2)^{  \colamntri}
        }
        \sbr{
            \pay_i\rbr{
                \subfunci(Z_i), \rbr{\subfunc_j(\varnothing)}_{j\neq i}
            }
        }
        &=
        \EEV{
            Z_i\sim N(\mu',\sigma^2)^{  \colamntri}
        }
        \sbr{
            \pay_i\rbr{
                \subfunci(Z_i), \rbr{\subfunc_j(\varnothing)}_{j\neq i}
            }
        }
        \\
        &\geq
        \inf_{\mu'\in\RR}
        \EEV{
            Z_i\sim N(\mu',\sigma^2)^{  \colamntri}
        }
        \sbr{
            \pay_i\rbr{
                \subfunci(Z_i), \rbr{\subfunc_j(\varnothing)}_{j\neq i}
            }
        }
        \\
        &=
        \inf_{\mu\in\RR}
        \EEV{
            \initdatai\sim N(\mu,\sigma^2)^{  \colamntri}
        }
        \sbr{
            \pay_i\rbr{
                \subfunci(\initdatai), \rbr{\subfunc_j(\varnothing)}_{j\neq i}
            }
        } .
    \end{align*}
    Because we have both matching upper and lower bounds, we conclude that $\forall \mu'\in\RR$
    \begin{align*}
        \inf_{\mu\in\RR}
        \EEV{
            \initdatai\sim N(\mu,\sigma^2)^{  \colamntri}
        }
        \sbr{
            \pay_i\rbr{
                \subfunci(\initdatai), \rbr{\subfunc_j(\varnothing)}_{j\neq i}
            }
        }
        = 
        \EEV{
            Z_i\sim N(\mu',\sigma^2)^{  \colamntri}
        }
        \sbr{
            \pay_i\rbr{
                \subfunci(Z_i), \rbr{\subfunc_j(\varnothing)}_{j\neq i}
            }
        } .
    \end{align*}
    Intuitively this says that the expected payment an agent receives is independent of which normal distribution they submit data from. But this means that agent $i$ can avoid the cost of data collection by submitting fabricated data with no reduction in payment. More formally, for any $\mu'\in\RR$, consider the strategy $\Tilde{\strat}_i:=
    (  0,\subfunc_{i,\mu'})$ . 
    % (0,\Tilde{\subfunc}_{i,\mu'})$ 
    % where $\Tilde{\subfunc}_{i,\mu'}(\cdot):=Z_i\sim N(\mu',\sigma^2)^{\datanum_i^\star}$. 
    Under this strategy agent $i$'s utility is 
    \begin{align*}
        \utilci\rbr{M, \rbr{\Tilde{\strat}_i, \strat_{-i}}}
        &=
        \inf_{\mu\in\RR}
        \EEV{
            Z_i\sim N(\mu',\sigma^2)^{  \colamntri}
        }
        \sbr{
            \pay_i\rbr{
                \subfunci(Z_i), \rbr{\subfunc_j(\varnothing)}_{j\neq i}
            }
        }
        \\
        &=
        \inf_{\mu\in\RR}
        \EEV{
            \initdatai\sim N(\mu,\sigma^2)^{  \colamntri}
        }
        \sbr{
            \pay_i\rbr{
                \subfunci(\initdatai), \rbr{\subfunc_j(\varnothing)}_{j\neq i}
            }
        }
        \\
        &>
        \inf_{\mu\in\RR}
        \EEV{
            \initdatai\sim N(\mu,\sigma^2)^{  \colamntri}
        }
        \sbr{
            \pay_i\rbr{
                \subfunci(\initdatai), \rbr{\subfunc_j(\varnothing)}_{j\neq i}
            }
        }
        -\costi  \colamntri
        % \\
        % &
        \,=
        \utilci\rbr{M, \rbr{\strat_i, \strat_{-i}}} .
    \end{align*}
    But this contradicts that $\strati$ is the best response to $\strat_{-i}$. Therefore, $  \colamntri=0$.
\end{proof}

%% file: Appendices/app_profit_max.tex
\section{ENVY-FREE PROFIT MAXIMIZATION}
\label{app:profit_maximization}

In this section, we will demonstrate how our mechanism can also be applied for envy-free profit maximization.
For this, let us define the profit of a mechanism $M$ under a strategy profile $\strat$: the total revenue---i.e., the sum of all payments received from buyers---minus the total cost incurred for data collection: $\sum_{j\in\buyers} \pricerj - \sum_{i\in\contributors}\costi\colamntri$.
Hence, the expected profit is given by,
\begin{align*}
   \hspace{-0.2in} \profit(M, s; \mu) = 
    \EE\Bigg[
        \sum_{j\in\buyers} \pricerj  -
            \sum_{i\in\contributors}\costi\colamntri
        \Bigg].
        \numberthis
        \label{eqn:profit}
\end{align*}

\subparahead{Envy-freeness}
First, note that it is not very hard to see our mechanism in Algorithm~\ref{alg:mechanism_c_known} already maximizes profit without additional constraints -- it maximizes the total welfare, while ensuring that the buyer's utility is $0$. However, as noted in~\S\ref{sec:mechanism_cost_known}, this pricing scheme may be ``unfair'' to some buyers, as all buyers pay different amounts for the same dataset.
A natural fairness requirement in this context is envy-freeness (EF)~\citep{varian1973equity,guruswami2005profit}, which states that no buyer should prefer another's data allocation and price to theirs.
We have formalized it as a fifth desideratum, in addition to the four given in~\S\ref{sec:mechdesignproblem}:

\begin{itemize}[leftmargin=0.25in]
\item[5.] \emph{Envy-free for the buyers (EFB):}
$\mechopt$ is envy-free for the buyers if no buyer prefers another buyer's dataset size selection rule and pricing rule to her own, when contributors follow $\stratopt$.
To state this formally, note that we can write a buyer's utility as a function of her own dataset size selection rule and pricing rule (with a slight abuse of notation):
\begin{align}
\utilbj(\mech, \strat) = \underbrace{\inf_{\mu\in\RR} \EE
                \left[ \valfuncj\left( 
           \left|\,\muhat(\receiveddata_j)
           \,-\, \mu\,\right| \right)
           \right]}_{:=f(\mmj, \strat)} - \underbrace{\sup_{\mu\in\RR} \EE\left[ \pricerj \right]}_{:=g(\pricej, \strat)}
    = \utilbj(\mmj, \priceoptj, \strat).
    \label{eqn:preefutilbj}
\end{align}
With this definition, the envy-freeness requirement can be stated as:
  $\utilbj(\mmoptj, \priceoptj, \stratopt) 
    \geq \utilbj(\mmoptk, \priceoptk, \stratopt)$ for all buyers $j,k\in\buyers$.
\end{itemize}

We will show how Algorithm~\ref{alg:mechanism_c_known} can also be used for envy-free profit maximization (Our ideas can also be adapted for envy-free welfare maximization).
First, in~\S\ref{sec:profitbaseline}, we establish a profit-maximization baseline similar to~\eqref{eqn:blwelfare} which assumes non-strategic contributors.  
Unfortunately, unlike in~\eqref{eqn:blwelfare}, there are no closed form solutions for this baseline;
hence,  in~\S\ref{sec:nonstrategicalgo}, we design algorithms for envy-free profit maximization with non-strategic contributors.
Then, in~\S\ref{sec:mechanism_profit_max}, we will combine the algorithm in~\S\ref{sec:nonstrategicalgo} with Algorithm~\ref{alg:mechanism_c_known} to 
derive mechanisms in the presence of strategic contributors.
% frame our mechanism design problem as approximating the profit-optimal baseline while accounting for strategic contributor behavior, we will build on the techniques in ~\S\ref{sec:nonstrategicalgo} when developing this mechanism.  
%Our ideas in~\S\ref{sec:buyerside} are not specific to mean estimation and apply to general data marketplaces. We begin by describing the environment.  

\subsection{A Baseline for Profit Maximization}
\label{sec:profitbaseline}

%\parahead{Envy-free Profit Maximization}

We construct the profit maximization baseline similar to~\eqref{eqn:blwelfare}, assuming that contributors are non-strategic and follow $\stratopt$. In this case, since buyers receive i.i.d.\ data, their valuations depend only on the data amounts $\{\mmrj\}_{j\in\buyers}$ (see~\eqref{eqn:valdataj}). Therefore, buyer $j$'s utility can be written as,
\[
\utilbj(\mech, \stratopt) = \utilbj(\mmj, \pricej,\stratopt) = \utilbj(\mmrj, \pricerj,\stratopt) = \valdataj(\mmrj) - \pricerj.
\]
% Without any constraints this problem is poorly defined, as the broker can simply ask all buyers to ``pay more'' to maximize profit. We will require the mechanism to satisfy the following two constraints:  
% \begin{enumerate}[leftmargin=0.2in]
%     \item \emph{Individually rational for buyers (IRB):}  
%     $\efscheme$ is IRB if $\utilbj(M,\stratopt) \geq 0$ for all buyers $j \in \buyers$.
%     \item \emph{Envy-free for buyers (EFB):}  
%     $\efscheme$ is EFB if no buyer prefers another's allocation and price over her own, i.e.
%     $\utilbj(M,\stratopt) = \valdataj(\mmrj) - \pricerj \geq \valdataj(\mmrk) - \pricer_k \quad \forall j,k \in \buyers$.
% \end{enumerate}  
% The first constraint ensures that buyers benefit from the market, as they would otherwise choose not to participate.  
% The second constraint ensures fairness, as otherwise the broker could sell all $\numdata$ points collected to all buyers and charge each buyer $\pricerj = \valitemj(\numdata)$.
% This maximizes profit while satisfying IRB;
% however, in this scheme different buyers would pay different prices for the same dataset, which would be unfair to buyers who have to pay more.  
Let us refer to all $\{\mmrj, \pricer_j\}_{j\in\buyers} \in (\RR\times \NN)^{|\buyers|}$ values which satisfy
IRB and EFB under $\stratopt$ as \emph{envy-free pricing schemes}.
Let $\efclass(N)$ denote all envy-free pricing schemes where the maximum amount of data sold is at most $N$, \ie $\max_j \mmrj \leq N$, and let $\efclass = \bigcup_{N\geq 0}\efclass(N)$.

\subparahead{Optimal baseline profit}
The optimal baseline profit $\blprofit$, subject to the IRB and EFB constraints, can therefore be written as shown in~\eqref{eqn:blprofitone}.
% Here, $\blprofit$ depends on the buyer valuations $\{\valdataj\}_{j\in\buyers}$ and contributor costs $\{\costi\}_{i\in\contributors}$, but we have suppressed this dependence for simplicity.
We have,
\vspace{-0.05in}
\begin{align*}
\numberthis \label{eqn:blprofitone}
    \blprofit &=
    \max_{M; \text{$M$ is envy-free}} \profit(M, \stratopt;\mu)
    = 
    \max_{\{\mmrj, \pricerj\}_{j} \in\efclass, \; \{\reqamntri\}_{i} \in\NN^{|\contributors|}  } \bigg(\; \sum_{j \in \buyers} \pricerj - \sum_{i \in \contributors} \costi \reqamntri\;\bigg).
\end{align*}
Above, the second step follows from the fact that since contributors will collect the requested amount and submit it truthfully, the broker's task reduces to finding \emph{scalar} values for dataset sizes and prices for buyers, and data collection amounts for contributors so as to maximize profit.
(Note that $\blprofit$ in~\eqref{eqn:blprofitone} for envy-free profit maximization is different to $\OPT$ in~\eqref{eqn:blwelfare} for welfare maximization.)
% \vspace{-0.1in}

\subparahead{Revenue-optimal envy-free pricing schemes}
To better understand $\blprofit$, we first note that as buyer valuations only depend on the amount of data they receive (and not which contributors collected the data), the broker's decisions
for buyers' dataset sizes $\{\mmrj\}_{j\in\contributors}$ and prices $\{\pricerj\}_{j\in\contributors}$ need to depend on $\{\reqamntri\}_{i\in\contributors}$ 
only through their sum $\numdata = \sum_{i\in\contributors}\reqamntri$, i.e. the total amount of data collected.
%Hence, let us assume that the seller has received $N$ total data points.
%The revenue $\ExpectREVmech$ for specific prices $\{\pricerj\}_{j\in\buyers}$, and 
The optimal revenue $\OPTREVmech$ under IRB, EFB constraints when the broker has received $\numdata$ total
data points can be written as shown below:

\begin{align*}
    %\ExpectREVmech\left(\{\pricerj\}_{j\in\buyers}\right) = {\sum_{j\in\buyers} \pricerj},
    %\hspace{0.4in}
    \OPTREVmech(\numdata) = \max_{\{\mmrj, \pricerj\}_{j\in\buyers} \in \efclass(\numdata)} \left(  \sum_{j\in\buyers} \pricerj\right).
    \numberthis
    \label{eqn:revoptmech}
\end{align*}
\vspace{-0.1in}
%\rcomment{$\OPTREVmech(N) $ or $\OPTREVmech(\numdata) $ here? }

\subparahead{From optimal revenue to optimal profit}
The following straightforward calculations lead to the following expression for the optimal envy-free profit $\blprofit$~\eqref{eqn:blprofitone}:
\begingroup
\allowdisplaybreaks
\begin{align*}
    \blprofit &= \max_{\{\mmrj, \pricerj\}_{j\in\buyers}\in\efclass, \{\reqamntri\}_{i}  } \left(\;  \sum_{j\in\buyers} \pricerj \;-\; \sum_{i\in\contributors} \costi \reqamntri  \;\right) \\
    &= \max_{\{\reqamntri\}_{i}  } \left(\max_{\{\mmrj, \pricerj\}_{j\in\buyers}\in\efclass(\sum_i \reqamntri)  }\;  \left(\sum_{j\in\buyers} \pricerj \right)\;-\; \sum_{i\in\contributors} \costi \reqamntri  \;\right) 
    \\
    &= \max_{\{\reqamntri\}_{i} } \left(\; \OPTREVmech\left(\sum_{i\in\contributors} \reqamntri\right) 
    \;-\; \sum_{i\in\contributors} \costi \reqamntri \; \right)
    \; = \;\OPTREVmech\left(\optnone\right) - \costone\optnone. 
    \numberthis
    \label{eqn:blprofittwo}
\end{align*}
\endgroup
where $\optnone = \argmax_{N\in\NN} \OPTREVmech(N) - \costone N$.
The last step uses the observation that agent $1$ has the smallest cost.
Intuitively, the broker should ask the cheapest agent to collect all the data.

Unfortunately, unlike in~\eqref{eqn:blwelfare}, there is no closed-form expressions for $\blprofit, \optnone$, which can be used as inputs for Algorithm~\ref{alg:mechanism_c_known}.
Next, we will show how these quantities can be approximated algorithmically.

\subsection{Designing Profit-optimal Envy-free Pricing Schemes with Non-strategic Contributors}
\label{sec:nonstrategicalgo}

Our first key insight here is that designing an optimal envy-free pricing scheme $\{(\mmrj, \pricerj)\}_{j\in\buyers}$ reduces to constructing a revenue-optimal pricing curve (Lemma~\ref{lem:equal_two_problem}).
% \ac{move up the lemma?}. 
Specifically, the seller posts a pricing curve $\itemprice: \{0,\dots,N\} \to [0,1]$ with $\itemprice(0) = 0$, and buyers choose their purchase quantities to maximize utility  
(in contrast, in our setting, the broker directly determines dataset sizes and prices for each buyer).  
When buyers have monotonic valuations, revenue-optimal pricing aligns with the ordered item pricing problem~\citep{chawla2022pricing}, allowing us to leverage existing algorithms.  
We first define the ordered item pricing problem below.
To simplify the exposition, we do so in the context of data pricing\footnote{In~\citet{chawla2022pricing}, the seller has $\numdata \in \mathbb{N}$ items, and a set of \emph{unit-demand} buyers, i.e. wish to purchase only a single item, with non-decreasing valuations $\valitemi: [\numdata] \to [0,1]$. While buyers may have different valuations for the same good, their preference ranking is the same. The seller should choose a pricing function $\itemprice: [\numdata] \to \mathbb{R}_+$ to maximize revenue.}.

%\parahead{Ordered item pricing\textnormal{~\citep{chawla2022pricing}}}
\begin{definition}\label{def:Ordered_item_problem}(Ordered item pricing~\citep{chawla2022pricing})
Suppose the broker has $\numdata \in \mathbb{N}$ data points and posts a pricing curve $\itemprice$.  
Buyers then decide how much to purchase based on $\itemprice$.  
A utility-maximizing buyer $j$ selects $m \in \{0, \dots, \numdata\}$ to maximize utility, $\valitemj(m) - \itemprice(m)$.  
If multiple values of $m$ yield the same maximum utility, the buyer chooses the largest dataset.  
Thus, buyer $j$ purchases $\mjp(\numdata, \itemprice)$ data points, where  
\begin{align*}
% \mjpset(\numdata, \itemprice)
\mjp(\numdata, \itemprice) \defeq
% \argmax_{m\in[\numdata]}
\max\left\{\argmax_{m\leq \numdata}
\big(\,\valitemj(m) - \itemprice(m) \,\big)\right\} .
% \hspace{0.4in}
% \mjp(\numdata, \itemprice) \defeq \argmax_{m\in\mjpset} m.
    \numberthis
    \label{eqn:buyerpurchasemodel}
\end{align*}
It follows that the revenue from buyer $j$ is $\itemprice(\mjp(\numdata, \itemprice))$.
Hence, the total revenue $\ExpectREV(\numdata, \itemprice)$ of a pricing curve
$\itemprice$ is as defined below.
The broker wishes  to find a pricing
function $\optitemprice$ which  maximizes the
cumulative revenue.
The optimal revenue $\OPTREV$ can be defined as:
% \vspace{-0.075in}
\begin{align*}
    \ExpectREV(\numdata, \itemprice) \defeq
        \textstyle\sum_{j\in\buyers} \itemprice\left(\mjp(\numdata, \itemprice)\right),
     %   \ac{p\rightarrow q?, m_i\rightarrow\Tilde{m_i}?}
    \hspace{0.5in}
    \OPTREV(\numdata) \defeq \max_\itemprice \;\ExpectREV(\numdata, \itemprice).
    \numberthis
    \label{eqn:revenue}
\end{align*}
\end{definition}
\subparahead{From ordered item pricing to revenue-optimal envy-free mechanisms}
Given a pricing curve $\itemprice$ with $\numdata$ points, it is straightforward to construct an envy-free pricing scheme $\{\mmrj,\pricerj\}_{j\in\buyers} \in\efclass(\numdata)$  while achieving the same revenue as $\itemprice$, i.e.  
$\ExpectREV(\numdata, \itemprice) = \sum_{j \in \buyers} \pricerj  $.  

To do so, we can set %$\efrevscheme = (\{\mmrj\}_{j \in \buyers}, \{\pricerj\}_{j \in \buyers})$ with 
$\mmrj = \mjp(\numdata, \itemprice)$ and $\pricerj = \itemprice(\mmrj)$ (see~\eqref{eqn:buyerpurchasemodel}).  
Since each buyer selects data to maximize her utility~\eqref{eqn:buyerpurchasemodel}, it follows that  
$\utilbj(\mmrj, \pricerj) \geq \utilbj(\mm, \itemprice(\mm))$ for every other $\mm\in[N]$,
and in particular, for all other buyers' dataset sizes $\{\mmk\}_{k\neq j}$. This ensures EFB.  
Moreover, as buying no data ($\mmrj = 0$ at $\pricerj = 0$) yields zero utility, IRB is also satisfied. 

This also implies that the revenue of the optimal envy-free pricing scheme is at least as large as the revenue of the optimal pricing curve, i.e. $\OPTREVmech(\numdata) \geq \OPTREV(\numdata)$ (see~\eqref{eqn:revoptmech} and~\eqref{eqn:revenue}).  

Our next lemma shows the converse, thus establishing an equivalence between optimizing ordered item pricing curves and envy-free pricing schemes.  
Thus, we have $\OPTREVmech(\numdata) = \OPTREV(\numdata)$.  
The proof  is provided in~\S\ref{sec:proofequaltwoproblem}.

\begin{lemma}%[ordered item pricing to revenue-optimal envy-free mechanisms] 
\label{lem:equal_two_problem}
Given the set of buyers $\buyers$, along with their valuations $\{\vali\}_{i\in\buyers}$.
Let $\numdata$ be the amount of data available to the broker.
For every envy-free pricing scheme $(\{\mmrj\}_{j \in \buyers}, \{\pricerj\}_{j \in \buyers}) \in\efclass(\numdata)$, there exists a 
pricing curve $\itemprice$ 
which achieves a higher revenue, \ie $\ExpectREV(\numdata, \itemprice) \geq \sum_{j\in \buyers} \pricerj$.
\end{lemma}
% \subparaheadwnogap{Designing mechanisms for revenue-optimal envy-free pricing}
% The above observations suggest the following procedure to design a mechanism:  
% Let $\algprice$ be the pricing curve returned by an ordered pricing algorithm with inputs $\buyers$, $N$, and approximation parameter $\eps$.  
% Define $\efrevscheme^A$ as the mechanism constructed using $\algprice$ as described above.  
% Then, $\ExpectREVmech(N, \efrevscheme^A) \geq \OPTREVmech(N) - \eps$.  

\insertTableApproxResults
\subparahead{Algorithms for ordered item pricing}
The ordered item pricing problem is known to be \textsf{NP}-hard~\citep{chawla2007algorithmic,hartline2005near}.  
As a result, prior work has developed approximation algorithms.  
An ordered item pricing algorithm, denoted by $A$, takes as input the buyer valuations $\{\valj\}_{j\in\buyers}$, the total amount of data $\numdata$, and an approximation parameter $\eps$.  
It outputs a pricing curve $\itemprice$ such that the cumulative revenue satisfies  
$\ExpectREV(\numdata, \itemprice) \geq \OPTREV(\numdata) - |\buyers|\bigO(\epsilon)$.
Table~\ref{tb:approxresults} summarizes algorithms from prior work under various assumptions on buyer valuations.  

%\vspace{0.5 cm}
\insertOIPAlgo
\parahead{Designing Profit-optimal Envy-free Mechanisms}
The above observations suggest the following procedure to design an approximately profit-optimal mechanism, by leveraging an algorithm for ordered item pricing.  
For any given $\numdata\in\NN$, we apply $A$ to determine an approximately optimal pricing curve.  
To maximize profit, we assign all data collection to contributor 1 and optimize over $\numdata$ to maximize revenue minus $\cost_1 \numdata$.  

Algorithm~\ref{alg:oipalgo} outlines this procedure.  
It takes as input buyer valuations $\{\valj\}_{j\in\buyers}$, contributor costs $\{\costi\}_{i\in\contributors}$, an ordered item pricing algorithm $A$, and an approximation parameter $\eps$.  
It returns a data quantity $\numdataA$, the corresponding pricing curve $\itemprice_{\numdataA}$, and dataset allocations $\{\mmjA\}_{j\in\buyers}$, all of which together guarantee an $\bigO(|\buyers|\eps)$-optimal revenue.  
Since buyer valuations are bounded in $[0,1]$, the maximum total buyer value is $|\buyers|$;
thus, it suffices to search for $N \leq |\buyers|/\cost_1$.  
While our brute-force approach is admittedly inefficient;
we leave it to future work to develop more computationally efficient procedures to find the optimal $\numdataA$.
The following lemma lower bounds the profit of Algorithm~\ref{alg:oipalgo}, with the proof given in Appendix~\ref{app:tech}.

\begin{restatable}{lemma}{Lempropfive}
    \label{lem:main_prop5}
    Suppose we execute Algorithm~\ref{alg:oipalgo} with an algorithm $A$ for ordered item pricing.
    Let $\numdataA$, $\{\mmjA\}_{j\in\buyers}$, $\{\pricejA\}_{j\in\buyers}$ be the returned values of Algorithm~\ref{alg:oipalgo}.
    Consider a mechanism $M$ where we let the contributors collectively collect $\numdataA$ amount of data (i.e., $\reqamntr_1 = \numdataA$, and $\reqamntr_i = 0$ for all $i \neq 1$). The broker then sell $\mmjA$ data to each buyer at price $\pricejA$, and issues payment $\payr_1  = \sum_j \pricejA$ to contributor 1 ($\payr_i = 0$, for all $i \neq 1$). Assume that the contributors follow $\stratopt$, we then have  $\profit(M,\stratopt;\mu) \geq \blprofit-|\buyers|\bigO(\epsilon)$.
    
   %Consider a mechanism obtained by choosing $\reqamntr_1=\numdataA$, $\reqamntr_i=0$ for all $i\in\contributors\backslash\{1\}$, $\mmrj = \mmjA$ for all $j\in\buyers$, and $\pricerj =\pricejA $. We then have $\efprofit(\efscheme) \geq \blprofit-|\buyers|\bigO(\epsilon)$.
\end{restatable}

\subsubsection{Proof of Lemma~\ref{lem:equal_two_problem}}
\label{sec:proofequaltwoproblem}

\begin{proof}
We need to show that for any envy-free pricing scheme \(\left\{(\mmrj, \pricerj)\right\}_{j \in \buyers}\), there exists a non-decreasing price curve \(\itemprice : \{0,1,\dots,\numdata\} \rightarrow [0, 1]\) that yields a revenue of at least $\sum_{j \in \buyers} \pricer_j$.
Without lose of generality, we assume $\mmr_1 \leq \mmr_2 \leq \dots \leq \mmr_{|\buyers|}$. We define \(\pricecurv\) as follows:
\begin{align*}
    {\pricecurv}(\mm) \defeq
    \begin{cases}
        \pricer_{1}, & \text{if } \mm \leq \mmr_1 \\
        \pricer_{2}, & \text{if } \mmr_1 < \mm \leq \mmr_2 \\
        \vdots \\
        %\pricecurv_{|\buyers| - 1}, & \text{if } \mm_{|\buyers| - 2} < \nn \leq \mm_{|\buyers| - 1} \\
        \pricer_{|\buyers|}, & \text{if } \mmr_{|\buyers| - 1} < \mm \leq \numdata
    \end{cases}
\end{align*}
This price function \(\pricecurv\) is non-decreasing and has at most $|\buyers|$ steps. By the purchase model of ordered item pricing \eqref{eqn:buyerpurchasemodel}, each buyer $j$ would purchase $\mjp(\numdata, \itemprice)$ data points under price curve $\pricecurv$, where
\[\mjp(\numdata, \itemprice) \defeq
% \argmax_{m\in[\numdata]}
\max\left\{\argmax_{m\leq \numdata}
\big(\,\valitemj(m) - \itemprice(m) \,\big)\right\}\]

It is sufficient to show that $\mjp(\numdata, \itemprice) \geq \mmrj$  holds for every buyer $j$. 
Then, as the price curve is non-decreasing,
the revenue from each buyer when using the pricing curve $q$ would be larger than her price in the envy-free pricing scheme. Therefore,
\[
\ExpectREV(\numdata, \itemprice) = \sum_{j \in \buyers}  \pricecurv(\mjp(\numdata, \itemprice)) \geq
\sum_{j \in \buyers}\pricecurv(\mmrj) = \sum_{j \in \buyers}\pricer_j .  
%= \ExpectREVmech(\{\pricer_j\}_{j\in\buyers}).
\]
which will complete the proof.

To show that $\mjp(\numdata, \itemprice) \geq \mmrj$,
let us first consider an arbitrary \(\mm \leq \mmr_{|\buyers|} \). For any \(\mm \leq \mmr_{|\buyers|} \), let \(k \in [|\buyers|]\) be such that \(\mmr_{k-1} < \mm \leq \mmr_{k}\). 
As valuations are non-decreasing we have,
\begin{align*}
    v_j(\mm) - \pricecurv(\mm) & \leq v_j(\mmr_{k}) - \pricecurv(\mm)
    = v_j(\mmr_{k}) - \pricecurv(\mmr_{k}) 
    \leq v_j(\mmr_{j}) - \pricecurv(\mmr_{j}).
\end{align*}
This implies that $\mmr_j \in \underset{\mm \leq \mmr_{|\buyers|}  }{\arg\max} \left( v_j(\mm) - \pricecurv(\mm) \right) $. 
%\kkcomment{Explain this last step a bit more. I had to think hard to understand what yo uwere saying.}
Finally, as $\mmr_{|\buyers|} \leq \numdata$,  we have
\vspace{-0.1in}
\[
\mmr_j \leq \max\Big(\underset{\mm \leq \mmr_{|\buyers|}  }{\arg\max} \left( v_j(\mm) - \pricecurv(\mm) \right)\Big) \leq \max \Big(\underset{\mm \leq \numdata}{\arg\max} \left( v_j(\mm) - \pricecurv(\mm) \right) \Big)
= \mjp(\numdata, \itemprice).
\]
\end{proof}

\subsubsection{Proof of Lemma~\ref{lem:main_prop5}}
%\Lempropfive*

\begin{proof}%[Proof of Lemma~\ref{lem:main_prop5}]
First, let $\itemprice' = A\rbr{\cbr{\valj}_{j \in \buyers},\optnone,\eps} $ denote the price curve obtained by applying $A$ to the buyers' valuations $\cbr{\valj}_{j \in \buyers}$, total amount of data $\optnone$, and the approximation parameter $\eps$. 
% We can upper bound $\blprofit$ as follows:
% Let $\itemprice' = A\rbr{\cbr{\valj}_{j \in \buyers},\optnone,\eps} $.
Using the definition of $\blprofit$ in equation~\eqref{eqn:blprofittwo} and the equivalence between envy-free pricing schemes and ordered item pricing curves (see Lemma~\ref{lem:equal_two_problem} and the discussion above it), we have
   \begin{equation*}     
      \blprofit 
      % & 
      = \;\OPTREVmech\left(\optnone\right) - \costone\optnone  = \; \OPTREV(\optnone) - \costone\optnone  .
      % \\ & \leq \; \ExpectREV(\optnone, \itemprice') + |\buyers|\bigO(\epsilon)- \costone\optnone  
      % \; \leq \; \ExpectREV(\reqamntoneopt, \itemprice_{\reqamntoneopt})- \cost_1 \reqamntoneopt  + |\buyers|\bigO(\epsilon) \\ & = \; \efprofit(\efscheme) + |\buyers|\bigO(\epsilon) 
   \end{equation*}
% Here, the first step follows from the definition of $\blprofit$ in equation~\eqref{eqn:blprofittwo}. The second step uses the equivalence between envy-free pricing schemes and ordered item pricing curves (see Lemma~\ref{lem:equal_two_problem} and the discussion above it).
Since algorithm $A$ returns a $|\buyers|\bigO(\eps)$ solution we get
\begin{align*}
    \OPTREV(\optnone) - \costone\optnone
    \leq
    \ExpectREV(\optnone, \itemprice') + |\buyers|\bigO(\epsilon)- \costone\optnone .
\end{align*}
% In the third step we have used the fact that the algorithm $A$ returns a $|\buyers|\bigO(\eps)$ solution.
Finally, from the definition of $\numdataA$: $  \numdataA = \arg\max_{N \in \cbr{1,2,\dots, \frac{|\buyers|}{\numdata}} } {\rm profit}_N $, since $\numdataA \leq \frac{|\buyers|}{\numdata}$, we have $ {\rm profit}_{\numdataA} =\ExpectREV(\numdataA, \itemprice_{\numdataA})- \cost_1 \numdataA  \geq  \ExpectREV(\optnone, \itemprice') - \cost_1 \optnone $. Therefore,
\begin{align*}
    \ExpectREV(\optnone, \itemprice') + |\buyers|\bigO(\epsilon)- \costone\optnone  
      \; &\leq \; \ExpectREV(\numdataA, \itemprice_{\numdataA})- \cost_1 \numdataA  + |\buyers|\bigO(\epsilon) 
      \\
      &= \; \efprofit(\efscheme, \stratopt;\mu) + |\buyers|\bigO(\epsilon)
\end{align*}
so we conclude that 
\begin{equation*}
    \efprofit(\efscheme,\stratopt;\mu) 
    \geq 
    \blprofit
    - 
    |\buyers|\bigO(\epsilon) .
\end{equation*}

\end{proof}

\subsection{Designing Profit-optimal Envy-free Mechanisms with Strategic Contributors}\label{sec:mechanism_profit_max}

Our mechanism with strategic contributors will execute Algorithm~\ref{alg:mechanism_c_known} with inputs obtained from the results of Algorithm~\ref{alg:oipalgo}. The following theorem describes the mechanism formally and states its theoretical properties.

\begin{restatable}{thm}{thmNIC} 
Let $\numdataA$, $\{\mmjA\}_{j\in\buyers}$, $\{\pricejA\}_{j\in\buyers}$ be the output of Algorithm~\ref{alg:oipalgo}.
Denote $\optprofitA = \sum_{j\in\buyers}\pricejA - \costone \numdataA$ be the profit of this envy-free pricing scheme.
Let $M$ be the mechanism in Algorithm~\ref{alg:mechanism_c_known}
executed with inputs
$\inputOPT = \optprofitA$, $\inputreqamnt = \numdataA$, 
$\buyersellj = \mmjA$, and $\buyerexppricej = \pricejA$ for all $j\in\buyers$.
Then, $M$ satisfies BB, IRC, ICC, IRB, and EFB.
    Moreover, its expected profit at the well-behaved NE satisfies
    \[
    \profit(\mechopt, \stratopt;\mu) \geq \blprofit - (\cost_2 - \cost_1) - |\buyers|\bigO(\eps),
    \]
\end{restatable}

% \begin{restatable}{thm}{thmNIC} 
% Let 

% Let $M$ be the mechanism in Algorithm~\ref{alg:mechanism_c_known}
% executed with input parameters $\OPT$ as in~\eqref{eqn:blprofittwo}, $\rbr{\reqamntoneopt,\cbr{\buyersellj}_{j \in \buyers},\cbr{\buyerexppricej}_{j \in \buyers} } = \rbr{\reqamntoneopt,\cbr{\mmjA}_{j \in \buyers}, \cbr{\pricejA}_{j \in \buyers}}$ as output of Algorithm~\ref{alg:oipalgo}.
% Then, $M$ satisfies BB, IRC, ICC, and IRB.
%     Moreover, its expected welfare at the well-behaved NE is $\welf(M, \stratopt) = \blprofit - (\cost_2 - \cost_1) - \left|  \buyers\right| \bigO(\eps) $.
% \end{restatable}

\begin{proof} The proofs of the BB, ICC, and IRC properties are identical to Theorem~\ref{thm:multi_thm}---we only change the input of Algorithm\ref{alg:mechanism_c_known} while maintaining the same payment format in equations (\ref{eq:buyer_pay}) and (\ref{eq:contrib_pay}). We will prove IRB, and EFB and then prove the lower bound on the profit.

\noindent
    \textbf{(IRB):}
As shown in (\ref{eqn:IRB}), $ \mathbb{E}\sbr{\pricerj} = \buyerexppricej $, since we let $\buyerexppricej = \pricejA$ be the output of Algorithm~\ref{alg:oipalgo}, we have $\mathbb{E}\sbr{\pricerj} =  \pricejA$. Then, combing with definition of buyer's utility (\ref{eqn:utilbone}), and buyers's value for i.i.d. data (\ref{eqn:valdataj}), buyer $j$'s utility under recommended strategy is 
       \begin{align*} \label{eqn:utilityj_stratopt}
         \utilbj(M, \stratopt)  \defeq & \; 
        \utilbj(\mmoptj, \priceoptj, \stratopt)  \defeq
        \valmechj(\mech, \stratopt)  -\sup_{\mu\in\RR} \EE_{M, \stratopt} \left[\pricerj \right]\\  =&  \; \valdataj(\mmjA)-  \sup_{\mu\in\RR} \EE_{M, \stratopt} \left[\pricerj \right]= \; \valdataj(\mmjA) - \pricejA \geq 0. \numberthis
    \end{align*}
    In the last step, recall that  $(\mmjA, \pricejA)_{j \in \buyers}$ are generated by Algorithm~\ref{alg:oipalgo}, combing with the buyer's purchase model in (\ref{eqn:buyerpurchasemodel}), we have $\valdataj(\mmjA)-  \pricejA \geq 0$ .

\noindent
    \textbf{(EFB):}
    Under $M$, following $\stratopt$ is envy-free for buyers.
% }
\begin{proof}
    To show that when contributors follow $\stratopt$, $\mech$ is envy-free for buyers, we must prove that $\utilbj(\mmoptj, \priceoptj, \stratopt) 
    \geq \utilbj(\mmoptk, \priceoptk, \stratopt)$, $\forall j,k\in\buyers$. As demonstrated in~\eqref{eqn:utilityj_stratopt}, we know $\utilbj(M, \stratopt)  \defeq \; 
        \utilbj(\mmoptj, \priceoptj, \stratopt) = \valdataj(\mmjA)- \pricejA \geq 0$. 
 Similarly, by the definition of buyers utility (\ref{eqn:utilbone}) and buyer's valuation for i.i.d. data (\ref{eqn:valdataj}), we have
       \begin{align*}           
        \utilbj(\mmoptk, \priceoptk, \stratopt)  \defeq  & \; \valdataj(\mmkA)-  \sup_{\mu\in\RR}\EE_{M, \stratopt}\cbr{\pricer_k  }  = \; \valdataj(\mmkA)- \pricekA 
 \leq \; \valdataj(\mmjA)- \pricejA .
    \end{align*}
The last step is by the definition of $\{\pricejA\}_{j \in \buyers}$ in Algorithm~\ref{alg:oipalgo}, and buyer's purchase model (\ref{eqn:buyerpurchasemodel}).
\end{proof} 

\noindent
    \textbf{Approximately Optimal Profit:}
    % \emph{
    Under $M$, when the contributors follow following $\stratopt$, the expected profit of buyers approximates the optimal profit upper bound (defined in Theorem~\ref{thm:neup}) with additive error $|\buyers|\bigO(\eps)$.
    % }
% }
\begin{proof}
    By definition of contributors' cumulative profit (\ref{eqn:profit}), 
    \begin{align*}
    \profit(M, \stratopt;\mu) & =  \EE_{M, \stratopt}\Bigg[
        \sum_{j\in\buyers} \pricerj -
            \sum_{i\in\contributors}\costi\reqamntri
        \Bigg]     =  \EE_{M, \stratopt}\Bigg[
        \sum_{j\in\buyers} \pricerj 
        \Bigg]  -\sum_{i\in\contributors}\costi\reqamntri  \\ & =   \sum_{j\in \buyers} \buyerexppricej -\sum_{i\in\contributors}\costi\reqamntri  = \sum_{j\in \buyers} \pricejA -\sum_{i\in\contributors}\costi\reqamntri \tag*{(By proof of \textbf{IRB}, equation (\ref{eqn:IRB}))}
\end{align*}
Recall that $\ExpectREV(\numdataA, \itemprice_{\numdataA}) =\sum_{j\in \buyers} \pricejA $
and that $\numdataA = \sum_i {\reqamntri} $, $\reqamntr_1 = \numdataA-1$, $\reqamntr_2 = 1$, we have 
   \begin{align*}
        \profit(M, \stratopt;\mu)   =& \;  
        \ExpectREV(\numdataA, \itemprice_{\numdataA}) -\cost_1 \rbr{\numdataA} -(\cost_2 - \cost_1)
         \geq  \; \ExpectREV(\optnone, \itemprice_{\optnone}) -\cost_1 \rbr{\optnone} -(\cost_2 - \cost_1)
       \\  \geq &  \; 
       \OPTREV
        \rbr{ 
            \optnone
        } 
        -  \cost_1 \rbr{\optnone} - |\buyers|\bigO(\eps)-  (\cost_2 -\cost_1)   
        =   \; \blprofit - (\cost_2 - \cost_1)- |\buyers|\bigO(\eps),
    \end{align*}
where the first inequality is by the optimality of $\numdataA$ in Algorithm~\ref{alg:oipalgo}, the second inequality is because algorithm $A$ returns a 
$|\buyers|\bigO(\eps)$ approximation, and the last step is by definition of $\OPT$ (\ref{eqn:blprofittwo}) and the fact that $\OPTREVmech(\optnone) = \OPTREV(\optnone)$.
\end{proof}

\end{proof}

%% file: Appendices/notation_table.tex
\section{Notations}
\label{sec:notation}

The following table contains the notations used in this paper. 
\begin{table}[H]
 \caption{Table of Notations.}
    \label{tab:notations}
    \centering
    \renewcommand{\arraystretch}{1.37}
    \begin{tabular}{ll}
        \hline
        Notation & Meaning \\
%        \hline & The notations below are described in Section \ref{sec:env}\\   
        \hline
        $\buyers$ & The set of buyers. \\
%        $\vali: \mathbb{N} \to [0,1]$ &  The valuation curve of each buyer $i \in \buyers$. \\        \hline
        $\contributors$  & The set of contributors. \\
        $\costi $ & The cost for each contributor to collect one data point. \\
       % \hline
        $\reqamntri $  &  Amount of data the broker asks each contributor to collect.
        \\ 
        %\hline
        $\reqamnti:\{\costi\}_{i\in\contributors} \to \reqamntri $  &  The function mapping the cost to $\reqamntri $.
        \\
        %\hline
        $\colamntri $  &  The amount of data each contributor actually collects.
        \\ 
        %\hline
        $\colamnti: (\reqamntri, \costi)\to  \colamntri $  &  The contributor $i$'s  strategy function to decide amount of data to collect.
        \\ 
        %\hline
         $ \initdatai $  &  The set of data contributor $i$ actually collects, containing $\colamntri $ i.i.d. points.
        \\ 
        %\hline
        $ \subdatai  $  &  The set of data each contributor reports.
        \\ 
        %\hline
        $ \subfunci  $  &  The contributors' data reporting function, may enable strategic behaviors (e.g., fabrication).
        \\ 
        %\hline
        $\strati = ( \colamnti, \subfunci) $ & Contributor $i$'s strategy profile.
        \\
        %\hline
        $\stratopti = (\identity, \identity)$ & Contributor is well-behaved: follow the instruction to collect $\reqamntri$ data and submit truthfully.
         \\ 
         %\hline
        $\mmrj$ & The amount of data points the broker allocates to each buyer $j \in \buyers$. 
        \\
        %\hline
        $\mmj$ & The broker's data allocation function to decide $\mmrj$. 
        \\ 
        %\hline
        $\pricerj$ & The broker charges buyer $j$ a price $\pricerj$ for $\mmrj $ data points.
        \\ 
        %\hline
        $\pricej$ & The broker's price function to decide $\pricerj$. 
        \\ 
        %\hline
         $\payri $  &  Payment that the broker assigns to each contributor.
        \\ 
        %\hline
        $ \payi $  &  The broker's payment function to decide how to reward contributor $i$ for data collection.
        \\ 
        %\hline
        $M  $ & Mechanism, $M= \left(\{\reqamnti\}_{i\in\contributors}, \{\mmj\}_{j \in \buyers}, \{\pricej\}_{j \in \buyers}, \{\payi\}_{i \in \contributors}\right)$. \\ 
        %\hline
         $\valfunc$  & Error-based buyer valuation \eqref{eqn:valerr}. 
        \\ 
        %\hline
        $\valmech$ &  Buyer valuation under a mechanism and a strategy profile (\ref{eqn:valb}).
        \\ 
        %\hline
        $\valdata$  & Buyer valuation under clean data \eqref{eqn:valdataj}. 
        \\ 
        %\hline
        $\utilbj(\mech, \stratopt)$& Buyer utility under mechanism $M$ and strategy $\strat $ \eqref{eqn:utilbone}. 
        \\ 
        %\hline
        % & The notations below are described in Section \ref{sec:mechdesignproblem} \\ \hline
        $\welf(M,\strat)$ & Welfare of a mechanism $M$ under a strategy profile $\strat$ (\ref{eqn:welfare}).
        \\ 
        %\hline
        $\OPT$ & Welfare (\ref{eqn:blwelfare}) or profit (\ref{eqn:blprofitone})  maximization baseline.
        \\ 
        %\hline
        $\reqamntoneopt$ & The amount of data to be collected to achieve welfare/profit-optimal welfare baseline $\OPT$.
        \\ 
        %\hline
        $\inputOPT$ & Input parameter in Algorithm~\ref{alg:mechanism_c_known}, denoting the welfare/profit-optimal baseline. 
        \\ 
        %\hline
        $\inputreqamnt$ & Input parameter in Algorithm~\ref{alg:mechanism_c_known}, denoting the total amount of data to be collected. 
        \\ %\hline
        $\buyersellj$ & Input parameter in Algorithm~\ref{alg:mechanism_c_known}, denoting the amount of data to sell to buyers.
        \\ %\hline
        $\buyerexppricej$ & Input parameter in Algorithm~\ref{alg:mechanism_c_known}, denoting the expected price to charge buyers.
        \\ %\hline
        % & The notations below are described in Section \ref{app:profit_maximization}. \\ \hline
        $\profit(M,\strat)$ & Profit of a mechanism $M$ under a strategy profile $\strat $  \eqref{eqn:profit}.
        \\ %\hline
        $A$ & Ordered item pricing algorithm, introduced in Section~\ref{sec:nonstrategicalgo}.
        \\ \hline
    \end{tabular}
    \vspace{0.5em}
\end{table}

%\kkcomment{Here, and whenevern possible, refer to the equation number where it was originally defined.}